\documentclass[journal]{IEEEtran}

\usepackage{amsmath,amssymb,amsfonts}
\usepackage[english]{babel}
\usepackage{algorithm,algorithmic}
\usepackage{graphicx}
\usepackage{textcomp}
\usepackage{color}
\usepackage{lipsum}
\usepackage{epstopdf}
\usepackage{amsopn}
\usepackage{amsthm}
\usepackage{enumerate}
\usepackage{cancel}
\usepackage{mathrsfs}
\usepackage{mathdots}
\usepackage{euscript}
\usepackage{amscd}
\usepackage{cite}
\usepackage{placeins}
\usepackage[latin1]{inputenc}
\usepackage{tikz}
\usetikzlibrary{snakes,arrows,shapes}

\newtheorem{theorem}{Theorem}
\newtheorem{corollary}{Corollary}
\newtheorem{lemma}{Lemma}
\newtheorem{proposition}{Proposition}

\theoremstyle{definition}

\theoremstyle{remark}\newtheorem{remark}{Remark}

\newcommand{\bmat}{\begin{bmatrix}}
\newcommand{\emat}{\end{bmatrix}}
\newcommand{\Frac}[2]{{\displaystyle\frac{#1}{#2}}}

\newcommand{\var}{\mathop{\rm var}}

\newcommand{\ltwonorm}[1]{\|#1\|_{\ell^2}}
\DeclareMathOperator{\diag}{diag}
\DeclareMathOperator{\trace}{tr}

\DeclareMathOperator{\rank}{rank}
\DeclareMathOperator{\supp}{supp}

\DeclareMathOperator*{\argmax}{argmax}
\DeclareMathOperator*{\argmin}{argmin}
\DeclareMathOperator{\symtoep}{SymToep}

\newcommand{\E}{{\mathbb E}}
\newcommand{\Rbb}{\mathbb R}

\newcommand{\Jbb}{\mathbb J}
\newcommand{\Cbb}{\mathbb C}

\newcommand{\Zbb}{\mathbb Z}

\newcommand{\xb}{\mathbf  x}
\newcommand{\yb}{\mathbf  y}
\newcommand{\sbf}{\mathbf  s}  
\newcommand{\zb}{\mathbf  z}
\newcommand{\wb}{\mathbf  w}
\newcommand{\vb}{\mathbf  v}

\newcommand{\ab}{\mathbf a}
\newcommand{\bb}{\mathbf  b}
\newcommand{\cb}{\mathbf  c}

\newcommand{\ub}{\mathbf  u}

\newcommand{\Cb}{\mathbf C}
\newcommand{\Db}{\mathbf D}

\newcommand{\Hb}{\mathbf H}

\newcommand{\Kb}{\mathbf K}

\newcommand{\Mb}{\mathbf M}

\newcommand{\Qb}{\mathbf Q}
\newcommand{\Rb}{\mathbf R}
\newcommand{\Sb}{\mathbf S}

\newcommand{\Xb}{\mathbf  X}

\newcommand{\varphib}{\boldsymbol{\varphi}}
\newcommand{\xib}{\boldsymbol{\xi}}

\newcommand{\omegab}{\boldsymbol{\omega}}

\newcommand{\Phib}{\boldsymbol{\Phi}}
\newcommand{\Sigmab}{\boldsymbol{\Sigma}}

\newcommand{\Bfrak}{\mathfrak{B}}

\newcommand{\Tfrak}{\mathfrak{T}}

\newcommand{\Fcal}{\mathcal{F}}
\newcommand{\Ycal}{\mathcal{Y}}
\newcommand{\Ucal}{\mathcal{U}}

\renewcommand{\d}{\mathrm{d}}
\renewcommand{\Re}{\mathrm{Re}}

\newcommand{\m}{\mathrm{m}}
\newcommand{\MAP}{\mathrm{MAP}}
\newcommand{\SNR}{\mathrm{SNR}}
\newcommand{\F}{\mathrm{F}}

\begin{document}

\title{An Empirical Bayes Approach to Frequency Estimation}

\author{Giorgio Picci and Bin~Zhu
\thanks{G.~Picci is with the Department of Information Engineering, University of Padova, Via Giovanni Gradenigo, 6b, 35131 Padova, Italy. B.~Zhu is with the School of Intelligent Systems Engineering, Sun Yat-sen University, Waihuan East Road 132, 510006 Guangzhou, China (email: \texttt{picci@dei.unipd.it, zhub26@mail.sysu.edu.cn}).}
\thanks{A preliminary version of this work was presented at the 17th European Control Conference (ECC 2019).}%
}


\maketitle

\begin{abstract}
In this paper we show that the classical problem of frequency estimation can be formulated and solved efficiently in an empirical Bayesian framework by assigning a uniform a priori probability distribution to the unknown frequency. We discover that the a posteriori covariance matrix of the signal model is the discrete-time counterpart of an operator whose eigenfunctions are the famous {\em prolate spheroidal wave functions}, introduced by Slepian and coworkers in the 1960's and widely studied in the signal processing literature although motivated by a different class of problems. The special structure of the covariance matrix is exploited to design an estimator for the hyperparameters of the prior distribution which is essentially linear, based on subspace identification. Bayesian analysis based on the estimated prior then shows that the estimated center-frequency is asymptotically coincident with the MAP estimate. This stochastic approach leads to consistent estimates, provides uncertainty bounds and may advantageously supersede standard parametric estimation methods which are based on iterative optimization algorithms of local nature. Simulations show that the approach is quite promising and seems to compare favorably with some classical methods.

\end{abstract}


\section{Introduction}

Frequency estimation is an old nonlinear problem  encountered   in many branches of science and engineering which has generated a huge literature. The survey of the literature up to 1993 in  \cite{Stoica-93} contains more than 300 titles. Since the literature on this problem is so large it is  impossible   to present    a reasonably complete summary in this introduction. For a general overview we shall just limit to refer to  the books   \cite{Stoica-M-05,Quinn-H-01,Kay-88} and  to  the references therein.

The most classical frequency estimation method   is via spectral analysis, based on the direct use of the periodogram which however   tends to produce nonconsistent estimates and must rely on ad hoc recombinations of partial spectral estimates (see e.g. \cite{Rice-R-88} and the comments in the introduction of  Thomson's paper \cite{Thomson-82}).     Research in this framework has nevertheless continued  and we should here at least point  to some recent interesting contributions such as \cite{Georgiou-00,Georgiou-01,BGL-00}.

Another   rather popular class of methods is based   on the so-called signal subspace decomposition.   The forerunner of signal subspace decomposition method (SSDM) is  Pisarenko harmonic decomposition, followed by  MUSIC, ESPRIT, and multiple signal classification methods. A survey of these methods can be found in the book   \cite{Stoica-M-05}. They are all based on linear algebra operations on the sample covariance matrix of the observed process and for this reason are quite popular. However in a way or another these methods rely on a rank estimation step  and on a (unavoidably approximate) rank-factorization of the sample covariance. This feature, in our opinion may generate some uncertainty on  their statistical properties, in particular consistency.

Accurate frequency estimation  has been mostly  approached in the literature by nonlinear optimization techniques, typically variants of Maximum Likelihood,   of which  a   remarkable example is  the early paper \cite{Nehorai-85}. Unfortunately,  because of nonconvexity, these methods are generally local and not guaranteed to yield a unique optimum.
Convex relaxation algorithms based on atomic norm minimization have  appeared   recently \cite{Yangetal-17},\cite{Zhuetal-17},\cite{Zhu-Wakin-17} but these methods rely on heavy regularization which in principle cannot produce unbiased estimates. A thorough  statistical analysis of these methods    still seems to be missing. 

\subsection*{New results}
In this paper we  follow  a Bayesin approach.  The underlying model is  the classical sum of harmonic oscillations corrupted by additive white noise,  whose frequencies are  modeled as   randomly varying parameters. Data are   modeled as trajectories of a process  whose   frequency  may   deviate sightly about an unknown  nominal value. It is then reasonable to   model frequency  as a random variable,  a  noisy  versions of  some  nominal frequency.

Bayesian estimation techniques for this model have been proposed in various places, e.g.\cite{Bretthorst-88,Bretthorst-97,Zacharias-etal-13,Turks-13,Dou-Hodgson-95,Djuric-Li-95} based on  various choices of the  prior distribution. Here we   propose an   approach based on the {\em Empirical  Bayes} philosophy,  inferring   from the observed data a family  of parametric prior distributions on the unknown frequencies. This approach    to frequency estimation seems to be new.  

The   parametric a priori density   is chosen as a uniform distribution on a small frequency range of unknown width, which   can be interpreted as  an a priori  confidence interval centered about  some  unknown  nominal frequencies. The  width and the relative center frequencies are the  {\em hyperparameters} of the prior which  are estimated from data. This simple model  seems to be a reasonable model for a variety of applications. Frequency variations on a  small bandwidth  could describe an experiment where one is   measuring the frequency shift of an oscillator (a function generator generating an AC waveform) with variable center frequency. That is, the  central frequency is unknown (random in $[-\pi,\pi]$) and also there is an unknown frequency shift of $[-W/2, W/2]$ radians/sample (which is also random and uniformly distributed). The random signal being  observed under additive white Gaussian  noise.\footnote{We thank one  reviewer for supplying this example.}

In this frame we show that the estimation of the hyperparameters can be approached by a simple efficient {\em subspace algorithm}. This in contrast with the standard marginal likelihood approach  as  considered for example  in \cite {Lazaro-Q-R-F-10,Aravkin-etal-12}. Our work uses more deeply the  structure of the data process and need not   involve optimization, going well  beyond the marginal likelihood approach.  For a  survey and some bibliography on Empirical Bayes methods  we refer to  \cite[p. 262]{Lehmann-C-98}, \cite{Efron-10,Efron-14,Chiuso-15,Petrone-etal-14,Aravkin-etal-12}. 
A general underlying motivation for the Empirical Bayes approach is that in some cases it has been proven   to yield a mean squared error (MSE),  which  can  even be smaller than maximum likelihood \cite{Reinsel-85,Yuanetal-16}.

Assuming a true model with a true  unknown  center frequency  hyperparameter, one can prove consistency of the subspace estimation method which  justifies our procedure in the framework of the traditional frequentist interpretation of the  hyperparameter. Later on,  we 
shall see  that   the (empirical) Bayesian MAP    frequency   estimate is   very close to (and in fact may asymptotically  coincide with)   the subspace centerfrequency  estimate. 

\subsection*{Relation with Prolate  Spheroidal Wave Functions}
Imposing the class of parametric uniform priors  leads to a simple   probabilistic structure of the signal. One ends up  by describing the observed signal as a special stationary process   named {\em bandlimited white noise } which has a flat power  spectrum within some finite bandwidth, whose  generation was first studied in the conference papers \cite{Favaro-P-15,Picci-Z-19}.  The remarkable fact is that  the covariance operator of these processes has isomorphic  properties  to those  uncovered  in the 60's and 70's by D.~Slepian and coworkers in a famous series of papers studying the energy concentration properties of time- and band-  limited signals, a completely different problem in a completely deterministic  context \cite{slepian1961prolateI,Slepian-P-61bis,Slepian-78,Landau-W-80}. The monograph \cite{hogan2011duration} is also  a good reference on this topic. We believe that an important contribution of this paper is to point out this stochastic interpretation and show its   usefulness in random  signal analysis. 
In  section \ref{CovProp} we  make contact with the classical works of David Slepian and colleagues.  In particular, here we elaborate on  the {\em bandpass} analogues of Prolate  Spheroidal Wave Functions  whose properties were still unknown, as mentioned  in a concluding remark in the paper   \cite{slepian1961prolateI}.

We discover that the whole theory  of bandlimited  time/frequancy  analysis of Slepian and co-workers, which for decades has  only been used for deterministic signal analysis, can be   transported to the stochastic setting allowing a deep understanding and a fine analysis  of the structure of the covariance of stationary signals with harmonic components. This has dramatic consequences.  For the first time our analysis  allows  a precise characterization of the finite-data approximation and truncation errors of the covariance kernel of the observed signal which is inherent in many covariance-based signal processing  methods of the literature.  Similar to Slepian's theory we discover that the eigenvalues of the covariance operator  decay abruptly to infinitesimal values (practically zero) after staying constant up to a certain {\em a priori computable} number, which can be identified as the numerical rank of the matrix. One can in fact get a rather precise estimate of the rank of a finitely-truncated covariance matrix and work with approximations of known precision. This was never suspected before and in all current literature, the  use of finite rank covariance approximations   to finite data sets  is assumed   without  much   of no  analysis of  the quality of  approximation.

In this  setting we can rigorously  justify the use of subspace methods based on finite rank purely-deterministic approximation of the process and its representation by state-space models.  

The proposed stochastic model embraces (in a Bayesian framework) the theoretical covariance structure  underlying many classical subspace methods used for frequency and DOA estimation such as   MUSIC, ESPRIT  and  descendants. In a sense our theory and results shed light on the foundations and approximation inherent in   these methods. In particular it allows a precise analysis of the finite-rank signal approximation  which is rarely addressed in the literature. As a result of  this analysis   a neat general proof of consistency can be provided.

More specifically, because of  the uniform frequency prior, the covariance of the observed  process turns out to be a function of the   {\em  modulated Sinc}-type, which in the special case  of nominal center frequency equal to  zero, has been well studied in the afore-cited literature. The key  property of the covariance  operator in question is that its eigenvalues decay extremely fast to zero for indices greater than an  \emph{a priori} computable number (the so-called {\em Slepian frequency} \cite{Khare-06}). This means that the eigenfunction expansion of the covariance kernel involves essentially only a {\em finite number of terms}. This key feature was  already evident and well-studied  in the classical deterministic literature  when the center frequency is zero but for non zero center frequencies  a  thorough understanding of the behavior of these modulated {\em Sinc }operators was posed as an open problem in \cite[p.~63]{slepian1961prolateI}. Later it was shown to hold for continuous-time modulated Sinc kernels  in \cite{landau1975szego,Landau-W-80} but the discrete-time case was left open. In this paper, we provide a  proof that modulated discrete-time kernels behave in a completely analogous way. This fact allows a direct and rather simple estimation of one hyperparameter  of the prior. The resulting center frequency estimate  is  computed by a subspace algorithm followed by a simple averaging process which seems to yield  very accurate and robust results, at least for a large enough sample size. This new  estimation  method is  expounded for signals with multiple unknown frequencies. 

\subsection*{Layout}

The paper  is organized as follows:

In Section \ref{sec:prob}, we formulate the Bayesian framework for the frequency estimation problem. We first  deal with signals with one hidden sinusoidal component but  the techniques and results are then  extended to treat signals with multiple harmonic components of unknown frequencies by assigning   them  non-overlapping rectangular (uniform)  prior distributions.
In this way the overall covariance kernel   becomes the sum of the individual covariances of   uncorrelated harmonic components.  Our technique can still  be applied and is somehow reminiscent of  Multiple Kernel methods   as in \cite{Hoffmann-S-S-08,Bach-L-J-04}. 

Then in Section \ref{CovProp}, we discuss the special structure of the signal covariance which is a discrete-time counterpart of the  {\em modulated Sinc kernel} class  discussed in the literature. We prove   the sharp decay property  of the eigenvalues using techniques inspired by  the continuous-time results from the literature. Then we illustrate our findings through a numerical example.

In Section \ref{sec:cov_est} we  exploit the covariance structure to propose an extremely  simple frequency estimate for signals with only one unknown frequency, which is only based on spectral data of the covariance. Note that because of non-ergodicity, consistent  estimation of the covariance data is a non-trivial issue. 

Section \ref{Subspace} attacks the main theme of the paper, namely estimation of multiple center-frequencies using a subspace method. By the finite rank  property one can use  a natural approximate state-space model of the  data.
 
Consistency of the subspace estimator is then discussed in Section \ref{Consist}.

Section \ref{Bayes} addresses the    MAP Bayesian estimator of the random frequency $\omegab$ based on the estimated prior discussed in  Section  \ref{Subspace}. 

In the following Section \ref{Simul2}, the method is applied to several test examples. As can be seen, the results are very encouraging.

At last, Section \ref{Conclusions} concludes the paper.

\subsection*{Notation and conventions}
Boldface symbols denote random quantities.
For a  square summable sequence $y$ of complex numbers, we take the definition of the discrete-time Fourier transform (DTFT) to be the following
\begin{equation*}
\begin{split}
\Fcal:\, \ell^2 & \to L^2[-\pi,\pi]\\
y & \mapsto \hat{y}(\omega):= \sum_{t\in\Zbb} y(t) e^{-i t \omega},
\end{split}
\end{equation*}
where the convergence of the Fourier series is understood in $L^2$ norm. The inverse transform is given by
\begin{equation*}
\Fcal^{-1}:\, \hat{y} \mapsto y(t) := \frac{1}{2\pi} \int_{-\pi}^{\pi} e^{i t \omega} \hat{y}(\omega) \d\omega.
\end{equation*}
The $\ell^2$ norm of $y$ is known as thefr energy of the signal.

The indicator function on a set $S\subset\Omega$ is defined as
\begin{equation*}
\chi_S(\omega) = \left\{ \begin{array}{ll}
1 & \textrm{for }\omega \in S,\\
0 & \textrm{for } \omega\in \Omega \setminus S.\\
\end{array} \right.
\end{equation*}

\section{Signal Model}
\label{sec:prob}

Consider the following signal model 
\begin{equation}\label{y_measurement}
\yb(t)=\xb(t)+\wb(t), \quad t\in \Zbb
\end{equation}
where $t$ represents  time,  $\xb$ is the sum of   random oscillatory components (a quasi periodic process), that is
\begin{equation}\label{x_multi_sinu}
\xb(t):= \sum_{\ell=1}^{\nu} \ab_\ell \cos (\omegab_\ell t) + \bb_\ell \sin (\omegab_\ell t),
\end{equation}
and $\wb$ is additive white noise. The angular  frequencies $\omegab_\ell$ are unknown but their number $\nu$   is fixed in advance. In addition  we shall require that:
\begin{itemize}
	\item the amplitude pairs $\ab_k,\bb_k$ are zero-mean pairwise and mutually uncorrelated   for all  $k$ and the two components $\ab_k,\bb_k$ have  equal variance:  $\sigma_k^2= \var [\ab_k]= \var [\bb_k], k=1,\ldots ,\nu$;
	\item each angular frequency $\omegab_\ell$ is a random variable taking values in the interval $[0,\pi]$, independent of the amplitudes;	
\item The noise $\wb(t)$ is assumed  white, zero-mean Gaussian,  stationary of variance $\sigma_{\wb}^2$, independent of everything else.   
\end{itemize}
We shall let $\omegab:= \bmat \omegab_1&\ldots& \omegab_\nu\emat^{\top}$ and denote by   $\ab,\,\bb$ two  similarly arranged amplitude vectors.  Note that the model is linear in $\ab,\,\bb$, and hence estimation of the amplitudes and their variance is just a standard linear estimation problem when the frequencies are  known. For this reason,  in this paper we shall mostly concentrate on the problem of frequency estimation.

Let us now introduce the Empirical Bayesian framework. We shall impose  that each component $\omegab_\ell$ of the random vector $\omegab$ follows a uniform distribution on the frequency band  $[\theta_\ell-W_\ell,\theta_\ell+W_\ell]$ such that the symmetrized sets w.r.t. the origin 
$$
S_\ell:=[\theta_\ell-W_\ell,\theta_\ell+W_\ell]\cup[-\theta_\ell-W_\ell,-\theta_\ell+W_\ell], \quad \ell=1,\dots,\nu
$$
do not overlap. 	For simplicity we shall assume that the assigned bandwidth is the same for different frequencies, i.e., $W_1=\cdots=W_\nu=W$. Here $0\leq\theta_\ell \leq\pi$ is called a {\em center-frequency} and $0\leq W\leq \pi$   the {\em bandwidth}.
In the literature, both $\theta$ and $W$ are called \emph{hyperparameters} of the {\em a priori} distribution for the frequency $\omegab$.

The stated  assumptions imply that for each fixed frequency value $\omega$ the $\nu$ components, say $\xb_\ell$, $\ell=1,\dots,\nu$ of the signal \eqref{x_multi_sinu}  are stationary uncorrelated processes. Hence the covariance function of the process $\yb$ for a fixed deterministic  $\omega$ has the form
\begin{equation}\label{OutCov}
\Sigma(t,s\mid \omega) := \E \left\{ \yb (t) \yb (s)\mid \omega \right\} = K(t,s\mid \omega) + \sigma_{\wb}^2 \, \delta (t,s)
\end{equation}
where $\delta(t,s)$ is the Kronecker symbol, and 
$$
K(t,s\mid \omega):=\sum_{\ell=1}^{\nu}\E \left\{ \xb_\ell(t) \xb_\ell (s) \mid \omega\right\}= \sum_{\ell=1}^{\nu} K_\ell(t,s\mid \omega)
$$
is the a priori conditional covariance of the signal $\xb$ given $\omegab= \omega$. To lighten the notation, we shall temporarily suppress the subscripts. The formulas below should be interpreted as holding for a generic index $\ell$.

By the model assumptions, the following computation is straightforward:
\begin{align}\label{ModSinc}
K(t,s\mid \omega) & = \E \left\{ \ab^2 \cos(\omega t) \cos(\omega s) + \ab\bb \cos(\omega t) \sin(\omega s) \right. \notag \\
 & \quad \left. + \ab\bb \sin(\omega t) \cos(\omega s) + \bb^2 \sin(\omega t) \sin(\omega s) \right\} \notag \\
 & = \sigma^2 \cos\omega \tau 
 \end{align}
where $\tau:=t-s$, and then computing the a posteriori covariance   by integrating the function w.r.t. the uniform prior density, one gets
\begin{align}\label{ModSinc1}
K(t,s) = \sigma^2 \, \E \left( \cos\omegab \tau \right)  
 &  = \sigma^2 \int_{\theta-W}^{\theta+W} \cos (\omega \tau) \frac{1}{2W} \d\omega \notag \\
 & = \sigma^2\cos (\theta \tau) \frac{\sin W\tau}{W\tau}.
\end{align}
Since the covariance function depends only on $\tau$, the signal $\xb$ is stationary, and so is $\yb$. In the following, we will write $K(\tau)$ in place of $K(t,s)$.

For $\theta=0$, the covariance function $K$ is the well-known {\em Sinc function}, which is the inverse Fourier transform of a rectangular function, namely
\begin{equation}
\sigma^2 \, \frac{\sin W\tau}{W\tau} = \frac{\sigma^2}{2W} \int_{-W}^{W} e^{i \omega \tau} \d\omega.
\end{equation}
It follows that    a zero-frequency  component of  the process $\xb$ must have   a uniform spectral density $\frac{\pi\sigma^2}{W}\chi_{[-W\,,W]}(\omega)$.
When $ W=\pi$, the process is just a usual stationary white noise of variance $\sigma^2$.
For $W<\pi$, the process $\xb$ is nontrivial, called a {\em bandlimited white noise} within the frequency band $[-W,\,W]$. In this case, it is a purely deterministic process with an absolutely continuous spectral distribution, since the logarithm of the density is obviously not integrable (see e.g., \cite[p.~144]{LPBook}).

In this paper, we are primarily interested in the case $\theta_{\ell} \neq 0$, for which we make the assumption that $|\theta_{\ell}|>W$, so that each support  set
\begin{equation}\label{set_S}
S:=[\theta-W,\theta+W]\cup[-\theta-W,-\theta+W]
\end{equation}
is composed of two  disjoint intervals symmetric with respect to the origin. Then the last expression in \eqref{ModSinc} can be rewritten as
\begin{equation}
\begin{split}
\sigma^2 \cos (\theta \tau) \frac{\sin W\tau}{W\tau} & = \frac{\sigma^2}{4W} \int_{-\pi}^{\pi} \cos (\omega \tau) \chi_S(\omega) \d\omega \\
 & = \frac{\pi\sigma^2}{2W} \int_{-\pi}^{\pi} e^{i\omega \tau} \chi_S(\omega) \frac{\d\omega}{2\pi} \\
\end{split} 
\end{equation}
where $\chi_S$ is the indicator function of $S$, and the second equality holds due to the symmetry of the integrand. From the above relation, we see that the spectral density of the process $\xb$ is now the sum of $\nu$ disjoint spectral terms, each of the form
$$
\phi_{\xb_{\ell}}(\omega)= \frac{\pi\sigma_{\ell}^2}{2W}\left( \chi_{[\theta_{\ell}-W,\;\theta_{\ell}+W]}+ \chi_{[-\theta_{\ell}-W,\;-\theta_{\ell}+W]}\right).
$$
The signal $\xb$ can therefore be described as a sum of independent deterministic carriers, each of angular frequency $\theta_{\ell}$, amplitude-modulated by a bandlimited white noise process described before.
For the same reason, the covariance function \eqref{ModSinc} has been called a {\em modulated sinc kernel} in \cite{Khare-06}, where it arises in a different context.

In practice we can only observe sample paths of finite length $N$ from the process $\yb$. For clarity of exposition, we shall now assume that $\nu=1$ and neglect the subscript $_\ell$ altogether. The generalization to multiple sinusoids, i.e., $\nu >1$, will be obvious.   Collect the observed random variables into a column vector, and in particular, let  $\Xb_N:=[\xb(t),\xb(t+1),\dots,\xb(t+N-1)]^\top$. Then consider the $N\times N$ covariance matrix 
\begin{equation}\label{Cov_matK}
\begin{split}
\Kb_N & :=\E\{\Xb_N\Xb_N^\top\} \\
 & = \begin{bmatrix}
 K(0)&K(1)&\cdots&K(N-1)\\
 K(1)&K(0)&\cdots& K(N-2)\\
 \vdots&\vdots&\ddots&\vdots\\
 K(N-1)&K(N-2)&\cdots&K(0)
 \end{bmatrix} .
\end{split}
\end{equation}
This symmetric Toeplitz structure of the covariance matrix comes from the fact that the process is stationary and real-valued. Similarly, we can define the $N\times N$ covariance matrix of the process $\yb$, say $\Sigmab_N$, and we have the relation
\begin{equation}\label{Cov_matSigma}
\Sigmab_N=\Kb_N+\sigma^2_{\wb} I_N.
\end{equation}
Analysis of the eigen-structure of $\Kb_N$ will be of great importance to our frequency estimation problem, and that will be the content of the next section.

\section{Properties of the Covariance Matrix}\label{CovProp}


In this section, we show that the covariance matrix \eqref{Cov_matK} also arises in a quadratic form which is the essential instrument for solving the energy concentration problem for discrete-time deterministic signals. In order to state the problem, we  first need to set up some notations.
Let $J$ be a set that is a union of a finite number of pair-wise disjoint closed subintervals of $[-\pi,\pi]$, e.g., a union of    sets like  $S$ in \eqref{set_S}. Define the band-limiting operator
\begin{equation}\label{oper_bandlimit}
\Bfrak:\, \ell^2 \to \ell^2, \quad y \mapsto \Fcal^{-1}[\chi_J \Fcal(y)]
\end{equation}
that corresponds to a bandpass filter with prescribed bandwidth $\{\omega \in J\}$.
Fix a positive integer $N$ and let
\begin{equation}\label{I_set}
I:=\{0,1,\dots,N-1\}.
\end{equation}
Define similarly the time-limiting operator
\begin{equation}\label{oper_timelimit}
\Tfrak:\, \ell^2 \to \ell^2, \quad y \mapsto \chi_I y,
\end{equation}
where $\chi_I$ is the indicator function in the time domain $\Zbb$.

The energy concentration problem that will be discussed in this section is
\begin{equation}\label{energy_concen}
\sup_{y \in \ell^2} \frac{\ltwonorm{\Bfrak \Tfrak y}^2}{\ltwonorm{y}^2}.
\end{equation}
Notice that the supremum can only be attained at a time-limited $y$, because the objective value of $\tilde{y}:=\Tfrak y$ is equal to $\ltwonorm{\Bfrak \tilde{y}}^2 / \ltwonorm{\tilde{y}}^2$ which is not less than that of $y$. Therefore, it is equivalent to consider the problem
\begin{equation}
\sup_{\substack{y \in \ell^2 \\ \supp(y) \subset I}} \frac{\ltwonorm{\Bfrak y}^2}{\ltwonorm{y}^2},
\end{equation}
where $\supp(\cdot)$ denotes the support of a function.
In other words, the aim is to find a time-limited signal whose energy is most concentrated in the frequency band $J$.

\subsection{The eigenvalue problem}

The impulse response of the ideal bandpass filter $\chi_J(\omega)$ is just the inverse Fourier transform
\begin{equation}\label{bandpass_imp_resp}
\rho(t) := \frac{1}{2\pi} \int_J e^{it\omega} \d\omega \quad t\in\Zbb.
\end{equation}
Observe that the function $\rho$ has the symmetry $\rho(-t)=\rho(t)^*$ where $z^*$ means the complex conjugate (transpose) of $z\in\Cbb$. 

According to the definitions \eqref{oper_bandlimit} and \eqref{oper_timelimit}, we have
\begin{equation}
\begin{split}
\Bfrak \Tfrak y & = \Fcal^{-1} \left[\chi_J(\omega) \Fcal(\Tfrak y)\right] \\
 & = \Fcal^{-1} \left[ \chi_J(\omega) \sum_{t=0}^{N-1} y(t) e^{-it\omega} \right] \\
 & = \rho*\Tfrak y
\end{split}
\end{equation}
where $*$ denotes convolution. It follows that
\begin{equation}
\begin{split}
\ltwonorm{\Bfrak \Tfrak y}^2 & = \sum_{t\in\Zbb} \left| \sum_{k=0}^{N-1} \rho(t-k) y(k) \right|^2 \\
 & = \sum_{j=0}^{N-1} y(j)^* \sum_{k=0}^{N-1} y(k) \sum_{t\in\Zbb} \rho(t-j)^* \rho(t-k).
\end{split}
\end{equation}
The last summation can be rewritten
\begin{equation}
 \begin{split}
 \sum_{t\in\Zbb} \rho(t-j)^* \rho(t-k)  
& =  \sum_{t\in\Zbb} \rho(j-t) \rho(t-k) \\
& = (\rho * x) (j),
 \end{split}
\end{equation}
where the sequence $x(t):=\rho(t-k)$ has Fourier transform   $\hat{x}(\omega)=e^{-ik\omega}\chi_J(\omega)$. The Fourier transform of $\rho*x$ is simply again $e^{-ik\omega}\chi_J(\omega)$. Hence the above sum is equal to $\rho(j-k)$, and we arrive at
\begin{equation}\label{norm_BTy}
\begin{split}
\ltwonorm{\Bfrak \Tfrak y}^2 & = \sum_{j=0}^{N-1} y(j)^* \sum_{k=0}^{N-1} y(k) \, \rho(j-k) \\
 & = \yb^* \Rb \yb,
\end{split}
\end{equation}
where $\yb= [\,y(0), y(1), \dots, y(N-1)\,]^\top$ is a slight abuse of notation, and
\begin{equation}\label{mat_R}
\Rb=\begin{bmatrix} \rho(0)&\rho(-1)&\cdots&\rho(-N+1)\\
\rho(1)&\rho(0)&\cdots& \rho(-N+2)\\
\vdots&\vdots&\ddots&\vdots\\
\rho(N-1)&\rho(N-2)&\cdots&\rho(0)
\end{bmatrix}.
\end{equation}
The matrix $\Rb$ has a Hermitian Toeplitz structure, and it is also positive definite because the quadratic form determines the energy of $\Bfrak \Tfrak y$. Notice that when the set $J$ is symmetric w.r.t.~the origin such as $S$ in \eqref{set_S}, then the integral in \eqref{bandpass_imp_resp} reduces to $\int_J \cos(t\omega) \d\omega$. In that case, $\rho$ is an even function of time, and the matrix $\Rb$ is real symmetric.

Now the objective functional in the energy concentration problem \eqref{energy_concen} is in fact equal to the Rayleigh quotient associated to $\Rb$. By the min-max theorem, the maximum of the objective is equal to the largest eigenvalue of $\Rb$, and it is attained when $\yb$ is  the corresponding eigenvector. It is obvious that the eigenvalues of $\Rb$ do not exceed $1$, simply because both $\Bfrak$ and $\Tfrak$ are projection operators.

\begin{remark}
Although it does not particularly interest us here, it is worth mentioning that the energy concentration problem \eqref{energy_concen} has a ``dual'' problem obtained by interchanging the two operators $\Bfrak$ and $\Tfrak$, namely
\begin{equation}\label{energy_concen_dual}
\sup_{y \in \ell^2} \frac{\ltwonorm{\Tfrak \Bfrak y}^2}{\ltwonorm{y}^2}.
\end{equation}
The problem \eqref{energy_concen_dual} is equivalent to determining the supremum of $\ltwonorm{\Tfrak y}^2$ over all band-limited signals subject to the constraint $\ltwonorm{y}=1$. By a standard variational argument using the Lagrange multiplier, one can conclude that the maximum of the dual objective is equal to the largest eigenvalue of a linear integral operator with a (modified) Dirichlet kernel. Moreover, following the lines in \cite[Section~5]{sengupta2012concentration}, it is not difficult to show that the eigenvalues of such an integral operator are identical to those of $\Rb$, and the corresponding eigenfunctions are related via the Fourier transform.
\end{remark}

\subsection{Asymptotic distribution of the eigenvalues}\label{subsec:asymp_distri_eigen}

We shall now allow the dimension of $\Rb$ to increase. In other words, the integer $N$ introduced by the set $I$ in \eqref{I_set} is considered as a variable tending to infiity. Let $\lambda_j(N)$ be the $j$-th eigenvalue (arranged in nonincreasing order) of $\Rb$. We know from the previous subsection that $0<\lambda_j(N)\leq 1$ for all $j=1,\dots,N$.
It also follows easily that
\begin{equation}\label{sum_eigen}
\sum_{j=1}^{N} \lambda_j(N) = \trace \Rb = N \rho(0) = \frac{\m(J)}{2\pi} N,
\end{equation}
where the notation $\m(\cdot)$ denotes the Lebesgue measure of a set.
Now for a real number $0<\gamma<1$, define $M(\gamma,N)$ to be the number of eigenvalues of $\Rb$ that are no less than $\gamma$. Again we have included the explicit dependence on the dimensional variable $N$. The next result is a first-order description of the asymptotic eigenvalue distribution of the matrix $\Rb$. The proof borrows techniques from \cite{landau1975szego} and can be found  in the appendix.

\begin{theorem}\label{thm_asymp_eigen}
It holds that
\begin{equation}\label{lim_asymp_eigen}
\lim_{N\to\infty} \frac{M(\gamma,N)}{N} = \frac{\m(J)}{2\pi} \,
\end{equation}
independent of $\gamma$.
\end{theorem}
\medskip
A more precise formula for the asymptotic expansion of the quantity $M(\gamma,N)$ is given  in \cite{Landau-W-80} for the continuous-time case. The second term in the  asymptotic expansion   is shown to be proportional to $ \log N $.  Slepian's asymptotic expressions for the eigenvalues, valid for $\theta=0$, are also reported  in \cite[p. 1059]{Thomson-82}.
 Although we believe that analogous discrete-time   estimates should hold, a formal proof is yet to be worked out.	For our problem of frequency estimation, Corollary \ref{cor_rank_R} below is anyway sufficient.

By choosing $\gamma$ arbitrarily close to $1$, an immediate consequence of the above theorem and formula \eqref{sum_eigen} is the following.
\begin{corollary}\label{cor_rank_R}
For $N\to \infty$, the matrix $\Rb$   has rank 
\begin{equation}\label{rank}
n=N\m(J)/2\pi, 
\end{equation}
and all the nonzero eigenvalues tend  to $1$.
\end{corollary}
The convergence is very fast since, as it is shown in  the proof of the Theorem, the matrix $\Rb$ has only $o( N)$ eigenvalues that are between $0$ and $1$, and for large sample size they can be reasonably neglected.

For $\nu =1$ the covariance matrix in \eqref{Cov_matK} is just a scalar multiple of $\Rb$ via $\Kb_N=\frac{\pi\sigma^2}{2W}\Rb_N$ (here for notational consistency we have added the subscript $_N$ to $\Rb$). Clearly, the constant factor only rescales the eigenvalues. In particular, the assertion on the rank in Corollary \ref{cor_rank_R} holds for $\Kb$. Below  we show some simulations of how the eigenvalues decay.

Fig.~\ref{fig:c40s} shows the behavior of the eigenvalues $\mu_k$  of the sinc kernel for   $N=1000, W/2\pi=0.02, \sigma^2=1$ which yields a rank approximately equal to $40$. 
We can clearly see  that for $n<40$ the eigenvalues are all equal to the same  constant   while for $n > 40$   the $\mu_k$'s very quickly decrease to  zero.

The behavior of the eigenvalues of $\Rb$  is the same  except  that the normalization    makes  the $\mu_k$  all practically  equal to one  for $k < n$. In order to  get the same normalization we just need to substitute $\mu_k$ with $2W \mu_k/2\pi$.

As for the modulated sinc kernel,  Fig.  \ref{fig:c80c} shows the eigenvalues of   a matrix $\Kb$ with the same values of $N$, $W$, and $\sigma^2$.
One sees that the eigenvalues have exactly the same behavior as those of the Sinc kernel.   Only  the value of $n$ such that for $k> n $, $\mu_k \simeq  0$ is now $4NW/2\pi=80$, i.e.,  twice the value of $n$ for the sinc kernel. Moreover, the amplitudes of  the eigenvalues   for $k < n$ are half of those of  the sinc kernel,  for equal values of $W$. This follows the from the symmetry of the spectrum and matches also the experimental findings of \cite{Khare-06}.

In order to get the largest eigenvalues of the modulated sinc kernel equal to one, a different normalization should be  made by substituting $\mu_k$ with $4W \mu_k/2\pi$. This   agrees with the matrix rescaling described above.

\begin{figure}[!h] 
\begin{centering}
  \includegraphics[width=0.45\textwidth]{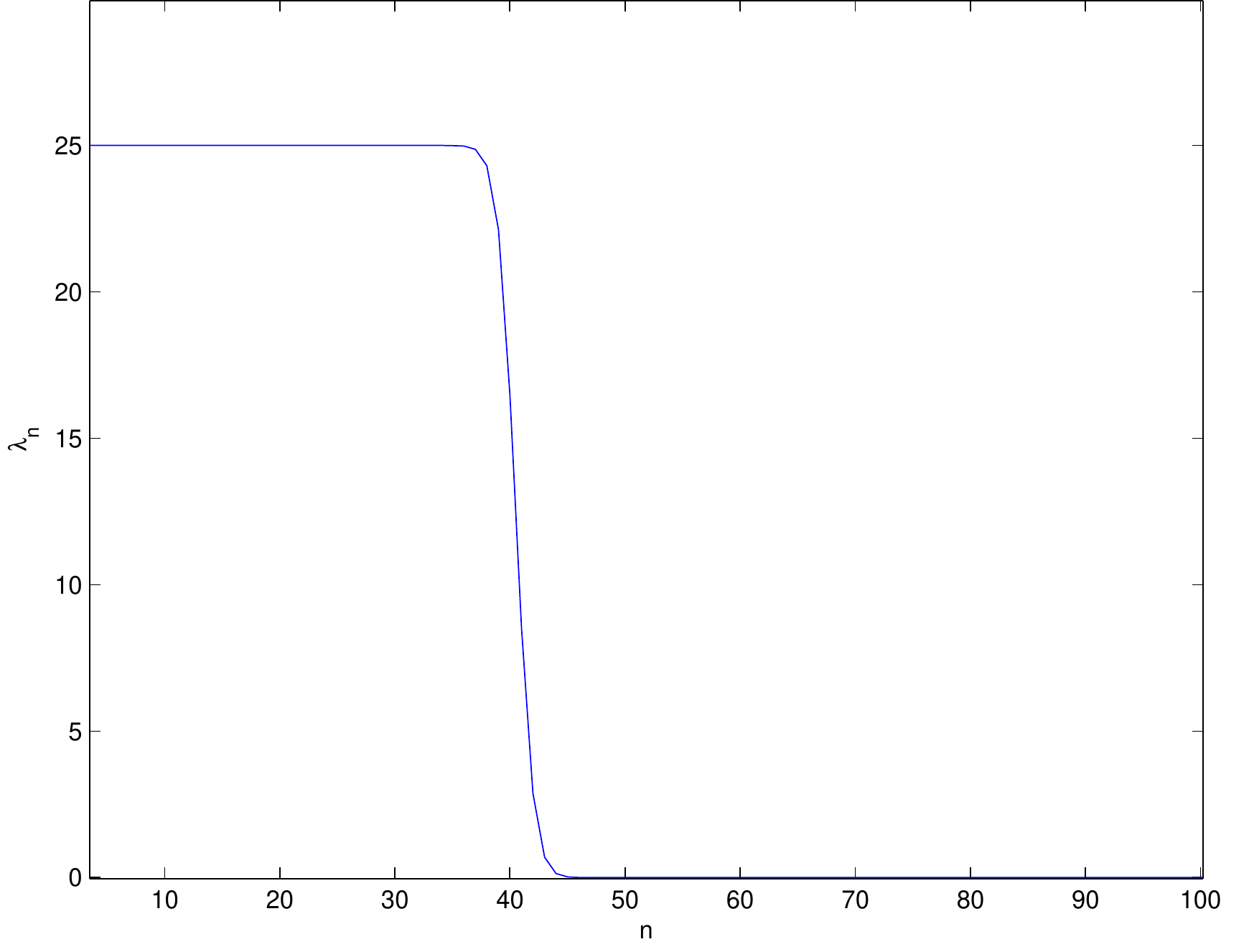} 
\caption{\em Eigenvalues of the sinc kernel covariance matrix, rank$\approx40$}
\label{fig:c40s}
\end{centering}
\end{figure}

\begin{figure}[!h]
\begin{centering}
  \includegraphics[width=0.45\textwidth]{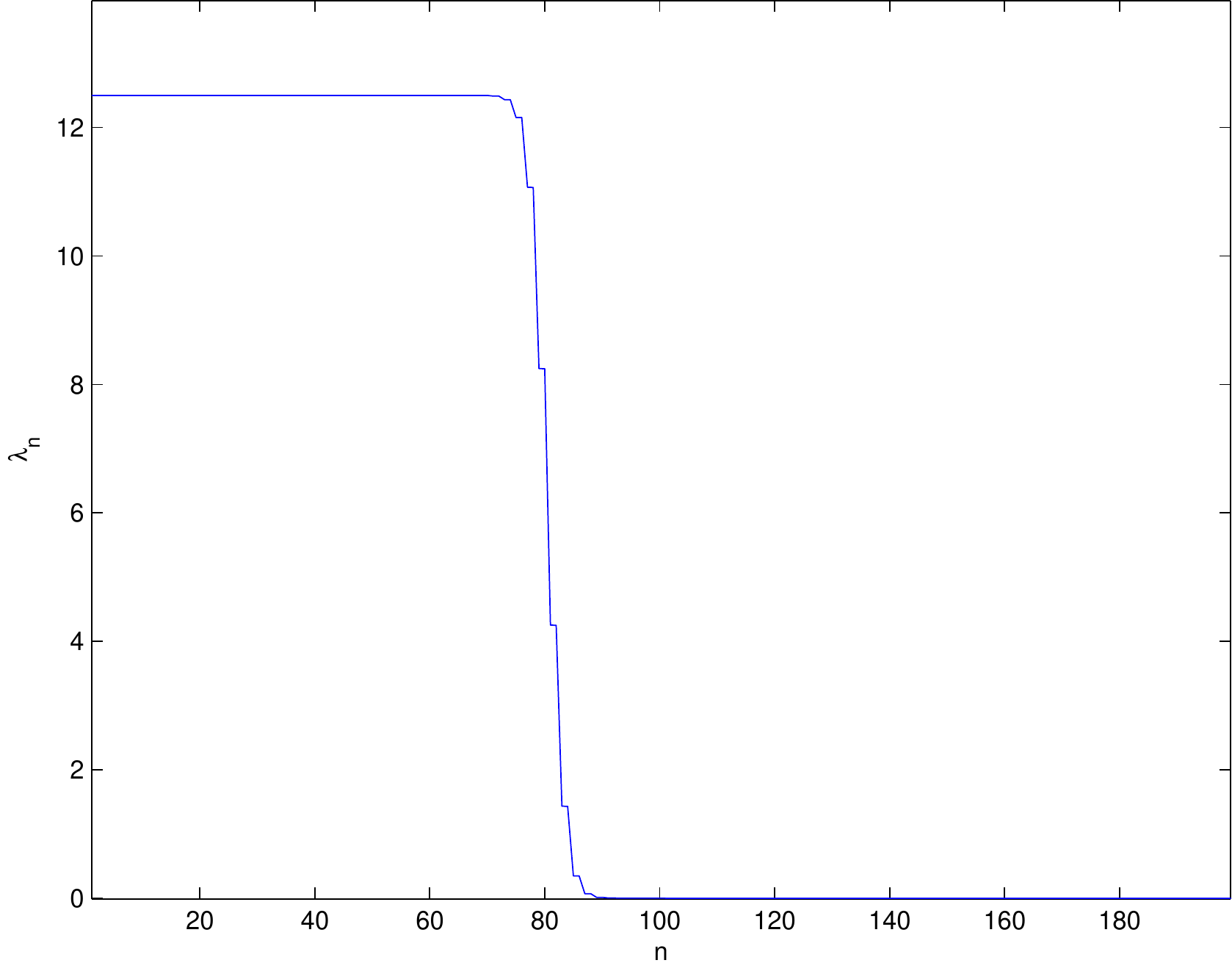}
\caption{\em Eigenvalues of the modulated sinc kernel covariance matrix, rank$\approx80$}
  \label{fig:c80c}
\end{centering}
\end{figure}

\section{Covariance Estimation}\label{sec:cov_est}

In the case of one hidden frequency, we have
$\rank \Kb_N\approx\frac{2W}{\pi}N$ according to Corollary \ref{cor_rank_R}. We can see that the bandwidth $W$ can be inferred from the rank information of the signal covariance matrix $\Kb$.
Since our measurements come from the process $\yb$, we start by estimating its covariance matrix $\Sigmab$.

A well-known difficulty in frequency estimation is that   stationary random processes with periodic components, even when the frequencies are  exactly known,  are not ergodic. Non-ergodicity means in particular that,  when the sample size   goes to infinity, the  limit of     the  process sample covariance is {\em sample dependent}, that is, the limit sample covariance  depends on the random amplitudes of its elementary oscillatory components (see e.g., \cite[pp. 105-109]{Soderstrom-S-89}). This lack of ergodicity is   even more serious when the frequency is random. For this reason, one-sample-path estimation runs into difficulty and  the standard  approach in  many practical situations is to consider estimation from cross-sectional or {\em panel data} (also called {\em snapshots}), as described in e.g., \cite{Hsiao-14} and e.g. done in DOA estimation.  Cross-sectional frequency data  can be the result of parallel measurements by multiple sensors  which is quite common for example in testing of turbo, and in general rotating machines, but also in many directional signal processing    and biomedical applications.

For the reasons above, we shall need to assume that our observed data consist of $L$ strings of sample  observations (snapshots), assumed for simplicity all of length $N$  :
\begin{equation}\label{y_meas_crs_sec}
y_k(t) = a_k \cos (\omega_k t) + b_k \sin (\omega_k t) + w_k(t),
\end{equation}
where $k=1,\ldots,L$, $t=1,\dots,N$, $(a_k,\, b_k)$ are sample determinations of the random variables $(\ab,\,\bb)$, and the frequencies $\omega_k$ are   sample determinations of the random variable $\omegab$ which is uniformly distributed  on the  fixed  interval $[\theta-W,\; \theta+W]$. We assume that noises of different cross sections are independent. Furthermore, we assume that the random samples $[a_k,b_k,\omega_k]$ come from  i.i.d.~copies of $[\ab,\bb,\omegab]$, then  the covariance matrix can be estimated by first subtracting the sample mean from the data, i.e.  
$$
\tilde y_k(t):= y_k(t) - \frac{1}{N} \sum_{t=1}^N y_k(t)
$$
and then doing a cross-sectional average
\begin{equation}\label{Sigma_hat}
\hat{\Sigmab}_{N,L} := \frac{1}{L} \sum_{k=1}^{L} \Ycal_k \Ycal_k^\top,
\end{equation}
where $\Ycal_k = \bmat \tilde y_k(1) & \cdots & \tilde y_k(N) \emat^\top$ is a column $N$-vector of centered data. The procedure is asymptotically equivalent (for $L\to \infty$) to first computing the standard (biased) covariance estimator  within each sample path \cite[Chapter~2]{Stoica-M-05},
$$
\hat{\sigma}_k(\tau) := \Frac{1}{N-\tau}\sum_{t=1}^{N-\tau} \tilde y_k(t+\tau) \tilde y_k(t),
$$
constructing the sample Toeplitz estimate
\begin{equation}\label{EstToepl}
\hat{\Sigma}_k :=  \symtoep \{\hat{\sigma}_k(0), \ldots,\hat{\sigma}_k(N-1)\}
\end{equation}
and then doing cross sectional average w.r.t.~$k$ to obtain  $\hat{\Sigmab}_{N,L} $ which is still symmetric-Toeplitz (here the subscript $N$ just refers to the dimension which is fixed).
By the strong law of large numbers, we have
\begin{equation}\label{converge_Sigma_hat}
\hat{\Sigmab}_{N,L} \to \Sigmab_{N} \text{ as } L\to\infty
\end{equation}
almost surely. Let $\hat{\lambda}_N$ be the smallest eigenvalue of $\hat{\Sigmab}_{N}$. Then given \eqref{Cov_matSigma} and Theorem~\ref{thm_asymp_eigen}, we have 
\begin{equation}\label{converge_min_eig}
\lim_{L,N\to\infty} \hat{\lambda}_N = \sigma^2_\wb.
\end{equation}
The limit here and those similar ones in the following are understood as first letting $L\to \infty$ and then $N\to\infty$.
In this sense  we are able to build a consistent estimator of the signal covariance matrix $\Kb_{N}$:
\begin{equation}
\hat{\Kb}_{N} := \hat{\Sigmab}_{N} - \hat{\lambda}_N I_N,
\end{equation}
 A consistent  estimator of the signal variance is given by
\begin{equation}
\hat{\sigma}_{\xb}^2:=\hat{\sigma}_{\yb}(0)-\hat{\lambda}_N,
\end{equation}
since we have $\sigma_{\yb}(0)=\sigma_{\xb}^2+\sigma_{\wb}^2$ by \eqref{OutCov}.

 Next, for $\varepsilon>0$ close to zero, the \emph{numerical rank} of $\Kb_{N}$ can be estimated using Theorem~\ref{thm_asymp_eigen} as
\begin{equation}
\rank(\Kb_{N})\simeq  M(\varepsilon,N) = \frac{2W}{\pi}N 
\end{equation}
with an approximation error which  roughly grows as $O(\log N)$.
In particular we have
	\begin{equation}\label{rank_est_W}
	\hat{W} := \frac{\pi}{2} \frac{\rank(\Kb_N)}{N} \to W
	\end{equation}
	when $N$ is large.
Unfortunately this  estimator of $W$ depends heavily on the estimate of the numerical rank whose computation is delicate and is not very reliable unless $N$ is very large. We shall comment on this  in the next subsection.

The relation for the scalar covariance of lag $1$
\begin{equation}
\sigma_{\xb}(1) = \sigma^2 \cos \theta \, \frac{\sin W}{W},
\end{equation}
could then be used to get a rough estimate of the center frequency:
\begin{equation}\label{est_f0}
\hat{\theta} := \arccos \left( \frac{\hat{\sigma}_{\xb}(1)}{\hat{\sigma}^2} \, \frac{\hat{W}}{\sin \hat{W}}\right),
\end{equation}
where $\hat{\sigma}_{\xb}(1)$ is an estimator of $\sigma_{\xb}(1)=\E \xb(t+1)\xb(t)$.

In the next section we shall describe a   more general reliable estimator based on the subspace  philosophy.

\begin{remark}
	The independence of  the cross sections, although often assumed in the literature, may seem quite strong. A more natural assumption could be  to require that the strings \eqref{y_meas_crs_sec} are sample  observations  of length $N$  from an {\em exchangeable} sequence of $N$-dimensional random vectors, $\{\yb_k\}_{k=1}^{L}$. For reasons of space this 
alternative viewpoint  will not be further pursued here.
\end{remark}


\section{A Subspace Approach to Hyperparameter Estimation}\label{Subspace}

Consider now the general measurement model \eqref{y_measurement}, with the signal $\xb$ consisting of multiple sinusoids as in \eqref{x_multi_sinu} satisfying all assumptions listed  in Sec.~\ref{sec:prob}. For simplicity, we shall assume that the amplitude variances are the same, $\sigma_1^2= \dots=\sigma_{\nu}^2= \sigma^2$. The covariance of $\yb$ can then  be computed similarly to that in Sec.~\ref{sec:prob}. We have 
\begin{equation}\label{Sigma_multi_freq}
\begin{split}
\Sigma(\tau) & = K(\tau) + \sigma_{\wb}^2 \delta(\tau,0) \\
 & = \sum_{\ell=1}^{\nu} \sigma^2 \, \E \left( \cos\omegab_\ell \tau \right) + \sigma_{\wb}^2 \delta(\tau,0) \\
 & = \sigma^2 \frac{\sin W\tau}{W\tau}\sum_{\ell=1}^{\nu} \cos\theta_\ell\tau + \sigma_{\wb}^2 \delta(\tau,0) \\
 & = \frac{\pi\sigma^2}{2W} \int_{-\pi}^{\pi} e^{i\omega \tau} \sum_{\ell=1}^{\nu} \chi_{S_\ell}(\omega) \frac{\d\omega}{2\pi} + \sigma_{\wb}^2 \delta(\tau,0). \\
\end{split}
\end{equation}
Under the  assumptions listed in Sec.~\ref{sec:prob}, the sum $\sum_{\ell=1}^{\nu} \chi_{S_\ell}(\omega)$ is the indicator function on the set $S:=\bigcup_{\ell=1}^{\nu} S_\ell$. From the integral expression for the covariance function, we see immediately that Corollary \ref{cor_rank_R} is applicable, and the asymptotic rank of $\Kb_N$ is now $\frac{2\nu W}{\pi}N$. A rank estimator for the bandwidth $W$ similar to \eqref{rank_est_W} can be used  since we have assumed that the supporting intervals for different frequencies have the same bandwidth. A more general situation with different $W$'s can also be dealt with but  it yields complicated  formulas and will not be discussed here. Next, we will concentrate on the estimation of the center frequency vector $\theta:=[\theta_1,\dots,\theta_\nu]^{\top}$.

\begin{remark}
	When the amplitudes $\sigma_1^2,\dots,\sigma_{\nu}^2$ are different, the spectral density of our signal is a sum of nonoverlapping rectangular functions and can always be written as a weighted sum of indicator functions. The assertion on the rank in Corollary~\ref{cor_rank_R} must still hold and a proof could be given based on Szeg{\"o}'s eingenvalue distribution theorem for Toeplitz matrices (see e.g.,\cite{gray2006toeplitz}).
\end{remark}

Efficient estimation of the hyperparameters can  be based on  maximum likelihood, assuming  Gaussian additive noise. See \cite[p. 429]{MacKay-92} for a general discussion of this point. The Gaussian likelihood function based on the $k$-th snapshot of  $N$ data can be written as (cf.~\cite{Hannan-D-88})
\begin{equation}
\begin{split}
l_k(\theta,W)= -\frac{N}{2}\log 2\pi - \frac{1}{2}\log \det \Sigmab(\theta, W) \\ -\frac{1}{2} \Ycal_k^\top \Sigmab(\theta, W)^{-1} \Ycal_k,  
\end{split}
\end{equation}
where $\Ycal_k$ is the vector introduced in \eqref{Sigma_hat}, $ \Sigmab(\theta, W) $ is the theoretical covariance matrix of $\Ycal_k$, with entries given in \eqref{Sigma_multi_freq} which do not depend on the index $k$. The first constant can be dropped from the objective function. By the independence of the sample paths, the log-likelihoods add to each other so that we end up with maximization of the function
\begin{equation}
l(\theta, W)= - \frac{L}{2}\log \det \Sigmab(\theta, W) -\sum_{k=1}^{L} \frac{1}{2} \Ycal_k^\top \Sigmab(\theta, W)^{-1} \Ycal_k
\end{equation}
with respect to  $\theta, W$. This leads to the well-know unique maximizer, see e.g. \cite[pp.~202--203]{Soderstrom-S-89}, for the covariance matrix
\begin{equation}\label{MomentEq}
\Sigmab(\theta, W)= \hat \Sigmab_{N,L}
\end{equation}
where $\hat \Sigmab_{N,L}$ is  defined in \eqref{EstToepl}.
Such an equation should be solved for the unknown hyperparameters $(\theta, W)$ appearing in the known structure \eqref{Sigma_multi_freq}. Note that this equation can be interpreted as resulting from the well-known  {\em method of moments} which is the theoretical basis of Subspace Methods \cite[Chapt. 13]{LPBook}.
Since the equation is nonlinear, one may think of setting up at the outset an iterative solution scheme. However, these numerical algorithms very often converge only locally. In fact, the likelihood function is nonconvex and contains many flat regions. Therefore, brute-force optimization seems to be a hard task.


We shall instead take advantage of the structure of the equation \eqref{MomentEq} to propose a {\em subspace-based} approach.
For a fixed and large enough $N$, we may and shall  here assume that the $N\times N$ covariance matrix $\Kb_N$ of the process $\xb$ has exactly rank $n:=\frac{2\nu W}{\pi}N$. As discussed in Subsection~\ref{subsec:asymp_distri_eigen}, for $N$ large this is a quite accurate approximation. In other words, we do a truncation in the spectral decomposition of the matrix $\Kb_N$, retaining the largest $n$ eigenvalues, namely
\begin{align}
\Kb_N=\Frac{\pi\sigma^2}{2W}\Rb_N =\Frac{\pi\sigma^2}{2W}\Qb_N \Db_N \Qb_N^\top \notag\\
\approx\Frac{\pi\sigma^2}{2W}\Qb_N \diag\{I_n,O_{N-n}\} \Qb_N^\top, \label{K_truncated}
\end{align}
where $O_m$ denotes the square all-zero matrix of size $m$. As before, the eigenvalues in the diagonal matrix $\Db_N$ are arranged in nonincreasing order.

\begin{proposition}
	For $N$ large enough, there are an  $n\times n$ matrix $A$ and an $n$-dimensional row vector    $c$ such that the random oscillatory signal $\xb$ can be represented by the system
	\begin{align}
	\xib(t+1)& = A \xib(t)   \label{stateeq} \\
	\xb(t) & = c\,\xib(t)   \label{outpeq}
	\end{align}
	where $\xib(t)= [\,\xi_1(t),\,\xi_2(t),\,\ldots, \xi_n(t)\,]^{\top}$ is an $n$-dimensional basis vector spanning the Hilbert space $\Hb(\xb)$ linearly generated by the $N$ random variables of the set $\{\xb(s)\,:\, t\geq s\geq t-N+1\} $. 
\end{proposition} 
\begin{proof}
	It is well-known that  a rank-deficient covariance matrix (of rank $n$) must necessarily be the covariance of a   purely deterministic process \cite[p. 138, 276]{LPBook}. When the total support of the spectrum $S=\bigcup_{\ell=1}^{\nu} S_\ell$ is a proper subset of $[-\pi,\pi]$, $\xb$ in \eqref{x_multi_sinu} is a purely deterministic process which can be represented by a deterministic linear recursion of order $n$ or equivalently, by a $n$-dimensional state-space model. Any such state-space representation for the process $\xb$ is of the form \eqref{stateeq}, \eqref{outpeq} 
	where $A$ can be chosen orthogonal so that  $A^{\top}=A^{-1}$.
\end{proof}

The output of \eqref{outpeq} has the expression $\xb(t) = cA^t \xib(0)$, from which we can compute the covariance function of the process as
$$
\sigma(t-s)= cA^t\E  \xib(0)\xib(0)^{\top} (A^{\top})^s c^{\top}= cA^t P\, A^{-s} c^{\top}.
$$
The matrix $P:=\E\xib(0)\xib(0)^{\top}$ satisfies a degenerate Lyapunov equation and commutes with $A$. Therefore, we have $\sigma(\tau) = cPA^{\tau} c^{\top}.$ The spectral density of $\xb$ is a sum of Dirac deltas. To see this, we first notice that since $A$ is orthogonal, its spectral decomposition can be written $A=T \Lambda T^*$  where $T$ is unitary and $\Lambda=\diag\{e^{i\varphi_1},\dots,e^{i\varphi_n}\}$ is a diagonal matrix of eigenvalues all having modulus $1$. The eigenvalues should come in conjugate pairs $e^{\pm i\varphi}$ if $\varphi\neq0,\pi$ due to the realness of $A$. The spectrum of the output process now follows:
\begin{equation}\label{spec_out_proc}
\begin{split}
\Phi_{\xb}(\omega) & = cP \Fcal(A^\tau) c^\top   = cPT \Fcal(\Lambda^\tau) T^*c^\top \\
 & = 2\pi\, cPT \diag\{ \delta(\omega-\varphi_1),\dots,\delta(\omega-
 \varphi_n)\} T^*c^\top, \\
\end{split}
\end{equation}
where the weights for the Dirac deltas are determined by the vectors $cPT$ and $T^*c^\top$. See also \cite[Eq.~(8.129)]{LPBook}.

Since the state-space realization will be constructed from the truncated covariance matrix \eqref{K_truncated}, its spectrum should approximate the true one, i.e., the indicator function on $S$  times a constant factor, in the sense that the supports of the Dirac deltas should be clustered in $S$. The center of each cluster, namely the average of the arguments $\varphi_k$ inside one cluster, is an estimate of the center frequency. Such an idea is also justified by the fact that the (approximate) eigenvalues of $\Kb_N$ do not depend on the center frequencies $\theta$. Hence the whole dependence on $\theta$ must be in $c$ and $A$.



Now the remaining point is how to obtain the parameters $c,\,A$ in the realization from the measurements of $\yb$. First, we estimate the rank of $\Kb_N$ using the technique in Subsection~\ref{subsec:sim_emp}.
Secondly, one can easily verify that the finite covariance matrix of $\xb$ in \eqref{outpeq} can be written as $\Kb_N = H_N P  H_N^\top$, where
\begin{equation}\label{mat_Hk}
H_k=\bmat c\\cA\\ \vdots\\cA^{k-1}\emat
\end{equation}
for a positive integer $k$.
This is in fact a rank $n$ factorization of $\Kb_N$. Notice that such a factorization is unique modulo the choice of basis in the state space and one can  always choose a basis such that $P$ is a diagonal matrix. In that case, we can compare with \eqref{K_truncated} and choose $H_N=\Qb_N(1:N,1:n)$ and $P$   just a constant multiple of the identity.
Thirdly, due to additive structure of the covariance matrix \eqref{Cov_matSigma}, $\Sigmab_N$ has the same eigenvectors as $\Kb_{N}$. Therefore, we can estimate the covariance matrix $\Sigmab_{N}$ using the scheme \eqref{EstToepl}, and extract the eigenvectors corresponding to the largest $n$ eigenvalues to compose $H_N$. Notice also that the variances of the signal and the noise do not affect $H_N$, and thus do not play a role in the later estimation.

The vector $c$ is simply the first row of $H_N$.
The matrix $A$ can be computed by a standard ``shift-invariance'' procedure of subspace identification. More precisely, for $k\leq N-1$, consider the matrix $H_k$ in \eqref{mat_Hk} and its one row shifted counterpart $\downarrow H_k:=H_N(2:k+1,:)$.
The dynamic matrix $A$ can be extracted by solving  the equation $\downarrow H_k =  H_k A$ in a least-squares sense. When $A$ is constrained to be orthogonal, this is the well-known ``orthogonal Procrustes problem''. In \cite[Subsec.~6.4.1]{Golub-VL}, it is reported that such a problem is well-posed, and can be solved using SVD.
A similar subspace method for oscillatory signals was proposed in \cite{Favaro-P-12a,Favaro-P-15}.

Given the cross sectional measurements \eqref{y_meas_crs_sec} of size $L\times N$, we summarize our algorithm below:

\begin{enumerate}
	\item Compute $\hat{\Sigmab}_N$, an estimate of the covariance matrix of $\yb$, using \eqref{EstToepl};
	\item Estimate the rank $n$ of the signal covariance matrix $\Kb_{N}$,
	and then estimate the bandwidth $W$ by \eqref{rank_est_W};
	\item Do eigen-decomposition to $\hat{\Sigmab}_N$, keep the largest $n$ eigenvalues, and call the $N\times n$ matrix of corresponding eigenvectors $H_N$;
	\item Let $k=N-1$, and solve the orthogonal Procrustes problem $\downarrow H_k =  H_k A$ for the orthogonal matrix $A$;
	\item Compute the eigenvalues of $A$, and extract their phase angles (between $-\pi$ and $\pi$);
	\item Run a clustering algorithm, e.g., $k$-means, on the phase angles, and take the centers of final clusters as estimates of the center frequencies.
\end{enumerate}
In the last step of this subspace algorithm, the center of each cluster may be obtained by simply taking the average of all the points in the cluster. This yields the estimate
\begin{equation}
\hat \theta_{\ell} = \Frac{1}{n_{\ell}}\,\sum_{k=1}^{n_{\ell}} \varphi_{k,\ell}\qquad \ell=1,\ldots,\nu
\end{equation}
where $n_{\ell}$ is the  number of  phase points in each cluster of positive phases.

\section{Consistency}\label{Consist}
Subspace methods for finite-dimensional models are essentially an instance of the method of
moments which is well-known in Statistics to be generically consistent under very mild assumptions. However, here the true  covariance
matrix is infinite-dimensional and the basic consistency analysis of moment estimation for
finitely parametrized models does not apply. In order to completely answer the convergence
question of the subspace-based estimator to the true frequency hyperparameter, one should
then combine the consistency property of subspace estimates which holds for the estimate of each finite dimensional approximate linear model (of fixed dimension), with the convergence, as the dimension of the covariance truncation tends to infinity, of the purely deterministic approximate process described previously to the a posteriori process which has a continuous spectrum. This  is a rather technical issue essentially centering on  symmetric Toeplitz spectral approximation which could not be   reported in this paper and is treated in a companion publication \cite{Picci-Z-20}. 

Consistency follows from a result of \cite{Picci-Z-20}  which establishes convergence (understood in a  weak sense) of the   line spectrum of the approximate model \eqref{stateeq},\eqref{outpeq}    to the continuous spectrum of the infinite Toeplitz covariance matrix. This  implies in particular that both the width and the centers of the discrete frequency  clusters must converge to the width and center of the corresponding intervals supporting the continuous spectrum which are indeed the true frequency hyperparameters. 

\section{Bayesian Estimation}  \label{Bayes}

Assume now that we have a consistent estimate  of the parameters of the prior, in particular of the center frequencies $\theta_{\ell}$. The question is what this estimate has to do with (say) the Bayesian Maximum A Posteriori (MAP) estimate\footnote{MAP is  known to be the best estimate  in a variety of norms.} of the random angular  frequency $\omegab$,    computed from  the relative posterior distribution.  Is there any reason why the MAP estimate should coincide, at least asymptotically,  with the center  frequencies of the prior?

  In the Subsection \ref{subsec:sim_Bayesian}
we  shall provide experimental evidence that    in our setting the inherent optimization problem leads to a MAP estimate of $\omegab$ which is  practically indistinguishable from the Empirical Bayes estimate of the center frequency $\theta$. This fact is verified experimentally but should be also evident from the theoretical analysis which follows.

The  MAP estimator of $\omegab$ is obtained by maximizing the $\log$ of the unnormalized posterior distribution of $\omegab$  given $N$ observations\footnote{The estimation from multiple snapshots data can be dealt with in a similar way even in case of  unequal measurement error variances.} $\yb:= \bmat y(t) & \ldots&y(t-N)\emat^{\top}$,
neglecting the denominator $p(y)$ which does not depend on the parameters. The prior for one frequency is 
$$
p(\omega\mid \theta_{\ell},W) = \frac{1}{2W}\,\chi_{[\theta_{\ell}- W,\;\theta_{\ell}+W]}
$$
and since the intervals do not overlap we have independence and  the overall prior of $\omegab$ is the product of the priors for each $\omega_{\ell}$ so that, recalling that the noise is Gaussian i.i.d. we have 
$$
\hat \omega^{\MAP}\!\!\!\!= \!\argmax_{\omega\in[0,\pi]^\nu}\left \{\! -\frac{1}{2\sigma_{\wb}^2}\, \| \yb\! - \!V(\omega ) \ub\|^2 +\!\!  \sum_{\ell} \log p(\omega\mid \theta_{\ell},W)\!  \right\}
$$
with $V(\omega) =  \bmat C(\omega) &S(\omega)\emat$ where 
\begin{align*}
C(\omega) &=  \bmat \cos \omega_1 &\ldots &\cos \omega_\nu \\ 
\vdots & \ddots& \vdots\\
\cos \omega_1N &  \ldots &\cos \omega_\nu N \emat := \bmat \cb_1(\omega_1)& \ldots&\cb_\nu(\omega_\nu)\emat  \\
S(\omega)& =  \bmat \sin \omega_1& \ldots & \sin \omega_\nu\\ \vdots & \ddots& \vdots\\
\sin \omega_1N& \ldots&\sin \omega_{\nu}N \emat := \bmat \sbf_1(\omega_1)& \ldots&\sbf_\nu(\omega_\nu)\emat 
\end{align*}
and $\ub=  \bmat a_1 & \ldots & a_{\nu}& b_1 &\ldots & b_{\nu} \emat^{\top}\,:= \bmat \ab &\bb\emat^{\top}$
which  could also  be written in complex form as $ \Re[\tilde V(\omega) \tilde\ub]$ where $ \tilde V(\omega)$ is the van der Monde matrix
$$
\tilde V(\omega)=\bmat e^{j \omega_1} &   \ldots &e^{j \omega_{\nu}}\\
\vdots & \ddots& \vdots \\
e^{j N\omega_1} &   \ldots &e^{j N\omega_{\nu}} \emat
$$
and $\tilde\ub:= \bmat  a_1-jb_1& \ldots &a_{\nu} -j b_{\nu} \emat^{\top}$. Since we are to compute real quantities this complex formulation does however not offer  substantial simplifications.

Now the log of the prior is $-\infty$ outside of the intervals $J_{\ell}:= [\theta_{\ell}- W,\;\theta_{\ell}+W] $ and equal to $ \log \frac{1}{(2 W)^{\nu}}$  inside (this is obviously true for each frequency and  true for the whole prior). Hence the MAP estimator of $\omegab$ can be found by solving the constrained minimization problem
\begin{align}
\hat \omega^{\MAP} &= \argmin_{\omega}\,\left \{ \frac{1}{2\sigma_{\wb}^2}\, \| \yb - V(\omega ) \ub\|^2  + \nu\log(2   W) \right\} \notag\\
\text{subject to} &:  \omega_{\ell} \in J_{\ell}\, \quad l=1,\ldots,\nu \label{BayesOpt}
\end{align}
Suppose that $\hat \theta_{\ell}, \ell=1,\ldots,\nu$ and $\hat W$ are our subspace estimates of the hyperparameters of the prior. Since these are consistent as discussed in the previous section, substituting these estimates for the true values leads to an asymptotically equivalent optimization problem. Here  $W$ appears as  a nuisance parameter which shall be fixed  to the  estimated width $\hat W$. The  Bayes MAP estimate of $ \omegab$   can then in principle be compute by minimizing  the quadratic  criterion $ \| \yb - V(\omega ) \ub\|^2$ subject to the fixed deterministic   constraint $J$: an hypercube in $\Rbb^{\nu}$  centered in $\hat \theta$ of edge length  $2\hat W$.

The minimization problem \eqref{BayesOpt} can then equivalently be interpreted as the {\em Maximum Likelihood} estimation of a {\em deterministic angular frquency} $\omega$  ranging on the bounded compact set $J$. On this set the likelihood function is smooth and, according to standard statistical theory, the estimate must be  consistent, that is converging for $N\to \infty$ to some "true value" $\omega_0$ which has generated the observations, and asymptotically efficient. 

For a finite data set  problems of the type \eqref{BayesOpt}  have in general several local minima.   However because of the bounded, compact, feasible set constraint   $\omega \in J$,  the solution must    stay in  a small neighborhood of the center frequency.  Also, the squared norm term in \eqref{BayesOpt} depends on $\nu$  sinusoidal functions of $\omega$ and hence, for small enough $W$ 's there are no equivalent values of the  frequency $\omega$ leading to the same value of the cost. The function has generically  a  unique minimum.

We  now propose an algorithm for the problem   \eqref{BayesOpt}  by using  the  a priori    estimate $\hat \theta$ as a starting point for a gradient descent and solve the problem by a local search algorithm about $\hat\theta$. Since  the subspace estimate, $\hat\theta$,  asymptotically  tends to the center frequency,  for large $N$ we are allowed to  identify $\theta$ with $\hat\theta$.

As a first preliminary step, solve a least squares problem minimizing $ \| \yb - V(\hat\theta ) \ub\|^2$ to get an estimate of the amplitude vector $\ub$\footnote{The estimate can also be justified based on a {\em noninformative prior} as in \cite{Zacharias-etal-13}.} and use  the estimated amplitude vector,
$$
 \hat\ub= [V(\hat\theta)^{\top}V(\hat\theta)]^{-1}V(\hat\theta)^{\top} \yb
 $$
in place of $\ub$  in the formulas. 
 
 Let  $\tilde\yb(\hat\theta):= \yb-V(\hat\theta)\hat\ub$ and introduce the deviation $\tilde\omega := \omega -\hat\theta$. The gradient of $V$ with respect to $\omega$ computed at $\hat\theta$, is an array of $ 2\nu$ rectangular  $N \times \nu$ gradient matrices of the form
 \begin{align}
\nabla V(\hat \theta) \!&=   \!\! \left [ \nabla_{\theta_1} \cb_1(\hat{\theta}_1),\! \ldots , \!\nabla_{\theta_\nu} \cb_\nu(\hat {\theta}_\nu),\!
   \nabla_{\theta_1} \sbf_1 (\hat{\theta}_1), \!\ldots  \!, \!\nabla_{\theta_\nu} \sbf_\nu(\hat \theta_\nu)\right ]   \label{FullGrad}
\end{align}
where each matrix entry has only the $k$-th column nonzero,  equal  (in Matlab notation) to
\begin{equation}\label{ColGrad}
\nabla_{\theta_k} \cb_k(\hat \theta_k)[:,k]= - D_N \sbf_k(\hat \theta_k), \ \nabla_{\theta_k} \sbf_k(\hat \theta_k)[:,k]=  D_N \cb_k(\hat\theta_k),\,
\end{equation}
where $D_N =\diag\{1,2,\ldots,N\}$.
 Hence $\nabla \{V(\hat \theta) \hat\ub\}$  turns out to be  a linear combination of these $ 2\nu$, $N \times \nu$ matrices, properly combined by the corresponding components of the vector $\ub \in \Rbb^{2\nu}$. By this operation the zero columns are superseded  and the linear combination leads  to a  $N \times \nu$ matrix  made by linearly combining  the $2\nu$ nonzero column vectors in \eqref{FullGrad}   to form a final matrix which we denote $\Mb(\hat\theta)$. For $\nu=1$ we have for example $\ub=\bmat a &b\emat^{\top}$ and 
 $$
 M(\hat\theta) = D_N(-\sbf(\hat\theta) a + \cb(\hat\theta) b )\in  \Rbb^{N\times 1}\,.
  $$ 
   With this gradient calculation established, we proceed to approximate \eqref{BayesOpt} by a   constrained local linear Least Squares minimization 
 \begin{align}
\min_{\tilde\omega}\,&\left \{  \| \tilde \yb - \Mb(\hat\theta) \,\tilde\omega \|^2   \right\} \notag\\
\text{subject to} &:  |\tilde\omega_\ell |\leq \hat W \quad \text{equivalent to} \;\omega_l \in J_\ell\, ,\label{BayesLin}
\end{align}  
for $\ell=1,\ldots,\nu$.
The solution can be refined  iteratively by an  algorithm of the form
\begin{align}
\tilde\omega(k+1) &=  [\Mb( \omega(k) )^{\top}\Mb( \omega(k) )]^{-1}	\times \notag \\
&\Mb( \omega(k) )^{\top}\tilde\yb(\omega(k))\,\;   k=1,2,,\ldots 
\end{align}
where at each step $\omega(k):= \tilde\omega(k) + \hat\theta$ is substituted back  in place of $\omega(k-1)$ or, initially, of $\hat \theta$ in the expression of the gradient. The scheme is initialized for $k=0$ setting $\omega(0)= \hat  \theta$ and then stopping when the difference $\tilde\omega(k+1)-\tilde\omega(k) = \omega(k+1)-\omega(k)$ becomes small enough. It    requires to check  at each step if $ |\tilde\omega_\ell |\leq \hat W$ otherwise  the estimator should be  re-initialized. Alternatively, we may try to keep $ \|\tilde\omega\|$ small by adding a ridge penalty term $\lambda(k)\|\tilde\omega(k)\|^2$ with $\lambda(k)\to 0$ for $k$ large for consistency, to the least squares formulation. This may in fact also  make the computation of the inverse better conditioned.

\begin{remark}
	The reasoning above can be extended to include multiple snapshots of data in a straightforward manner. Since the conditional likelihood function for each snapshot multiplies given the hidden frequencies, the squared-norm term in the objective function of \eqref{BayesOpt} becomes $\|\Ycal-V(\omega)\Ucal\|^2_{\F}$, where $\Ycal$ and $\Ucal$ are matrices whose columns are the data and the amplitude vectors, respectively, and the subscript $_\F$ denotes the Frobenius norm.
	A similar linearization scheme can be devised to solve the enlarged optimization problem.\hfill $\Box$
\end{remark}

\section{Simulations}\label{Simul2}

In this section, we provide simulation evidence showing that the subspace algorithm described at the end of Section \ref{Subspace}   works quite well in the case of one or two  hidden frequencies. Simulations comparing with the MAP estimate will also be shown.

In the second step of the subspace algorithm,  in order to compute an estimate of the bandwidth $W$ using \eqref{rank_est_W} we need to estimate the asymptotic rank of the signal covariance matrix. It turns out that such a rank estimation task can be tricky if we are given (relatively) a small number of samples. This point will be discussed in the next subsection.

\subsection{The bandwidth estimator}\label{subsec:sim_W}



In the first example, we compare the decay property of eigenvalue sequence of the estimated covariance matrix with the theoretical behavior as shown in Figs.~\ref{fig:c40s} and \ref{fig:c80c} in the case of two hidden frequencies.
The measurements \eqref{y_meas_crs_sec} are generated with $\ab,\bb$ with uniform distribution $U[-1.3813,1.3813]$ and $\omegab$ drawn from the uniform distribution in $[\theta-W,\theta+W]$ with the hyperparameters $\theta=[\theta_1,\theta_2]=2\pi\times [0.3145, 0.4201]$ and $W=2\pi\times0.0465$. \footnote{These numbers come from one trial in the Monte-Carlo simulations.}
The signal length $N$ and the number of snapshots $L$ are both equal to $100$.
The additive noise is i.i.d.~Gaussian with
variance $\sigma^2_{\wb}$. The signal-to-noise ratio (SNR) defined as $20\log_{10}(\sigma/\sigma_{\wb})$ has a value of $15$ dB.
In Fig.~\ref{fig:eig_y}, we report the eigenvalues of the estimated covariance matrix \eqref{EstToepl}.
\begin{figure}[!h]
	\begin{centering}
		\includegraphics[width=0.5\textwidth]{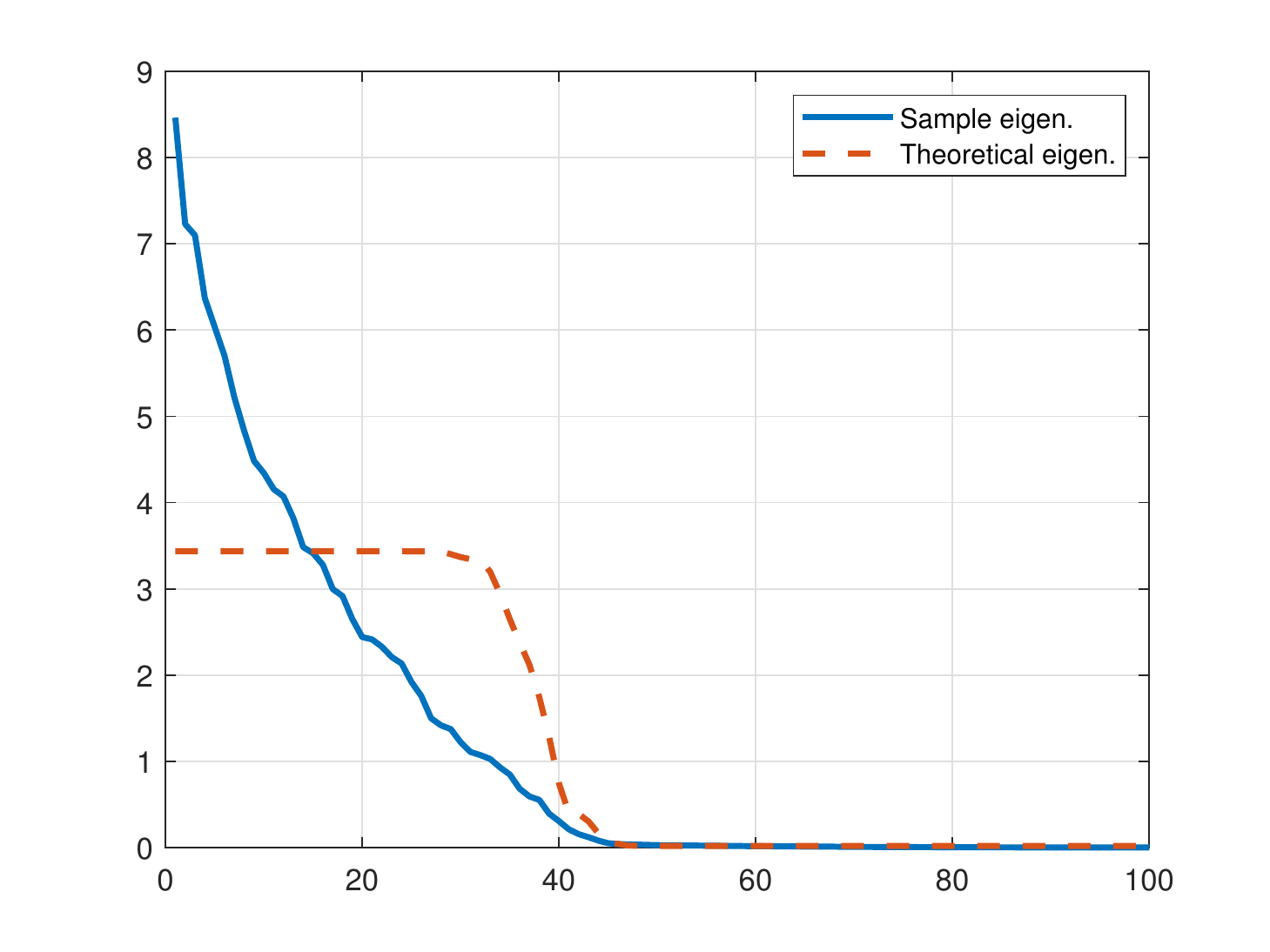}
		\caption{Eigenvalues of the theoretical and estimated covariance matrices with $N=L=100$. }
		\label{fig:eig_y}
	\end{centering}
\end{figure}

By comparison with the eigenvalues of the theoretical covariance matrix (red dashed line), we can see a significant distortion in the large eigenvalues due to the slow convergence of the estimator \eqref{EstToepl}.
However, the flat regions of two eigen-sequences still overlap nicely.
Inspired by such an observation, we propose an ad-hoc scheme: 
replace $\rank(\Kb_{N})$ in \eqref{rank_est_W} with the index maximizing the following ratio
\begin{equation}\label{rank_substitute}
\argmax_{k\in\{1,\dots,N-1\}} \frac{\lambda^2_k(\hat{\Sigmab}_{N})}{\lambda^2_{k+1}(\hat{\Sigmab}_{N})},
\end{equation}
where $\lambda_k(\hat{\Sigmab}_{N})$ denotes the $k$-th eigenvalue of the estimated covariance matrix $\hat{\Sigmab}_{N}$ arranged in nonincreasing order.
Intuitively, the maximum should be attained at the beginning of the flat region in the eigen-plot.

Next, we do a Monte-Carlo simulation to test our idea. In each trial, the hyperparameters are generated randomly. More precisely, first the bandwidth $W$ is drawn from the uniform distribution in $2\pi\times[0.01,0.05]$, and then the center frequencies $\theta_1$ and $\theta_2$ are drawn from $U \,[W,\pi-W]$ such that $|\theta_1-\theta_2|>2W$ so that the supporting intervals for the two frequencies do not overlap.
Given $L$ independent measurement sequences of length $N$, the covariance matrix is estimated through \eqref{EstToepl}, and then the rank is computed via \eqref{rank_substitute}, which gives an estimate of $W$ by \eqref{rank_est_W}. The \emph{relative} estimation error of $\hat{W}$ is defined by the ratio $(\hat{W}-W)/W$. Notice that we have not taken the absolute value of the numerator because we want to show that the scheme \eqref{rank_substitute} tends to \emph{overestimate} the rank of the signal covariance matrix. This feature is important in practice since the estimated rank determines the eigen-truncation performed in the subspace algorithm (Step 3). Clearly, we want to retain the eigenvectors of the covariance matrix corresponding to large eigenvalues. Hence, an underestimation of the rank should be avoided since otherwise, useful information about the spectral content of the signal could be lost.

Each Monte-Carlo simulation consists of $1000$ trials.
In the first experiment, we fix $N=L=100$ and estimate the bandwidth $W$, or equivalently the numerical rank of the signal covariance matrix, under different SNRs. 
In Fig.~\ref{fig:W_est_SNR}, the relative errors of $\hat{W}$ are depicted using the boxplot. We see from the box on the right that a low SNR results in an underestimate of the rank which is undesirable for the subsequent estimation of the band centers. The overall error is not small mainly because we have a poor estimate of the covariance matrix given the number of available samples (see Fig.~\ref{fig:eig_y}). However, we want to emphasize that the estimation of $W$ is a separate problem, and a large error here does not propagate to the estimation of the center frequencies. As we will see in the next subsection, the center frequencies can be estimated quite accurately given a rough estimate of the bandwidth.


In the second experiment, we fix the $\SNR=15$ dB and estimate $W$ as both $N$ and $L$ change while keeping $N=L$.
The result is depicted in Fig.~\ref{fig:W_est_N=M}. One can see that as $N=L$ increases, the estimates become more and more accurate.

\begin{figure}[!h]
	\begin{centering}
		\includegraphics[width=0.5\textwidth]{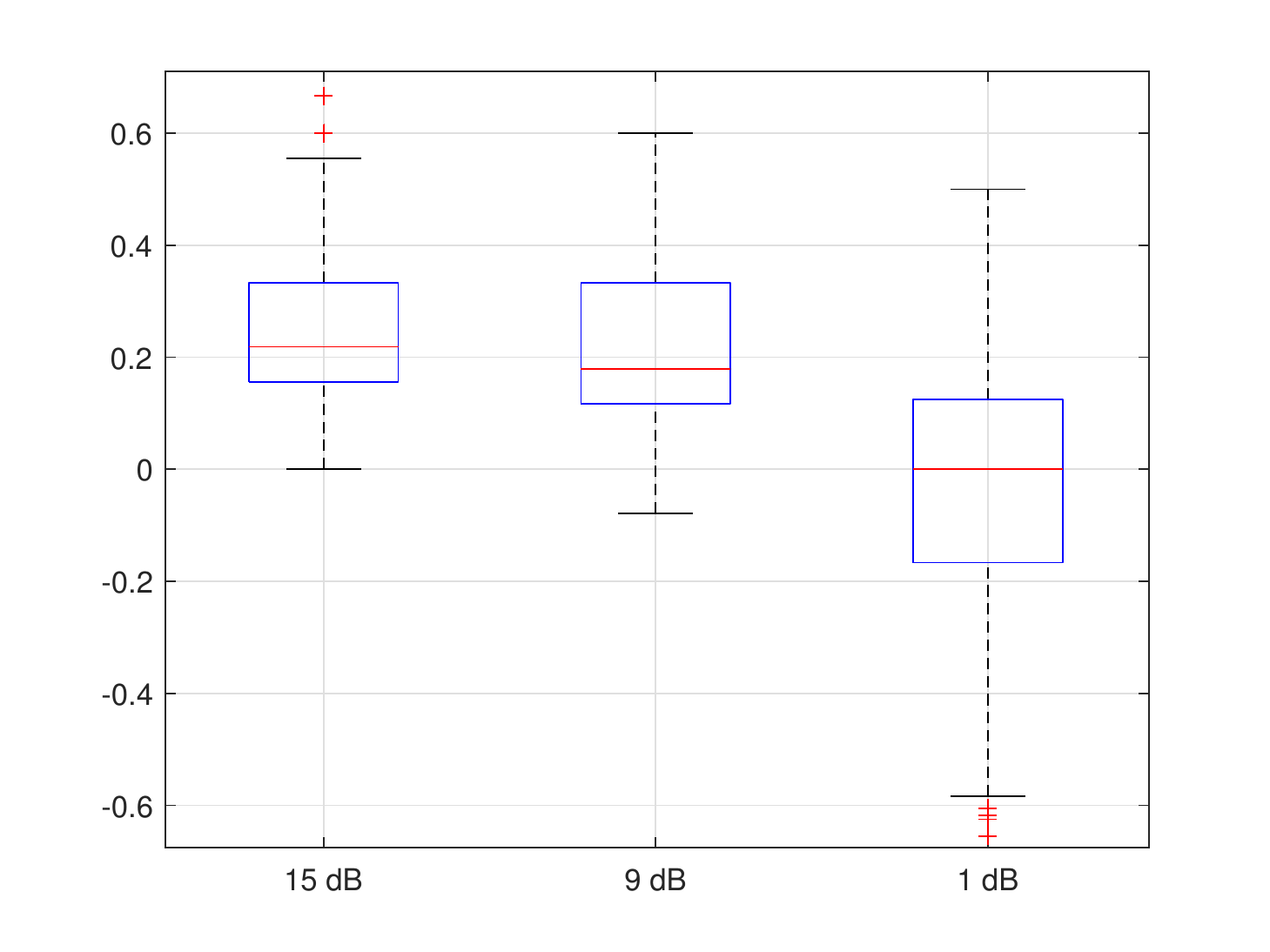}
		\caption{Relative estimation errors of the bandwidth $W$ versus the SNR with $N=L=100$.}
		\label{fig:W_est_SNR}
	\end{centering}
\end{figure}


\begin{figure}[!h]
	\begin{centering}
		\includegraphics[width=0.5\textwidth]{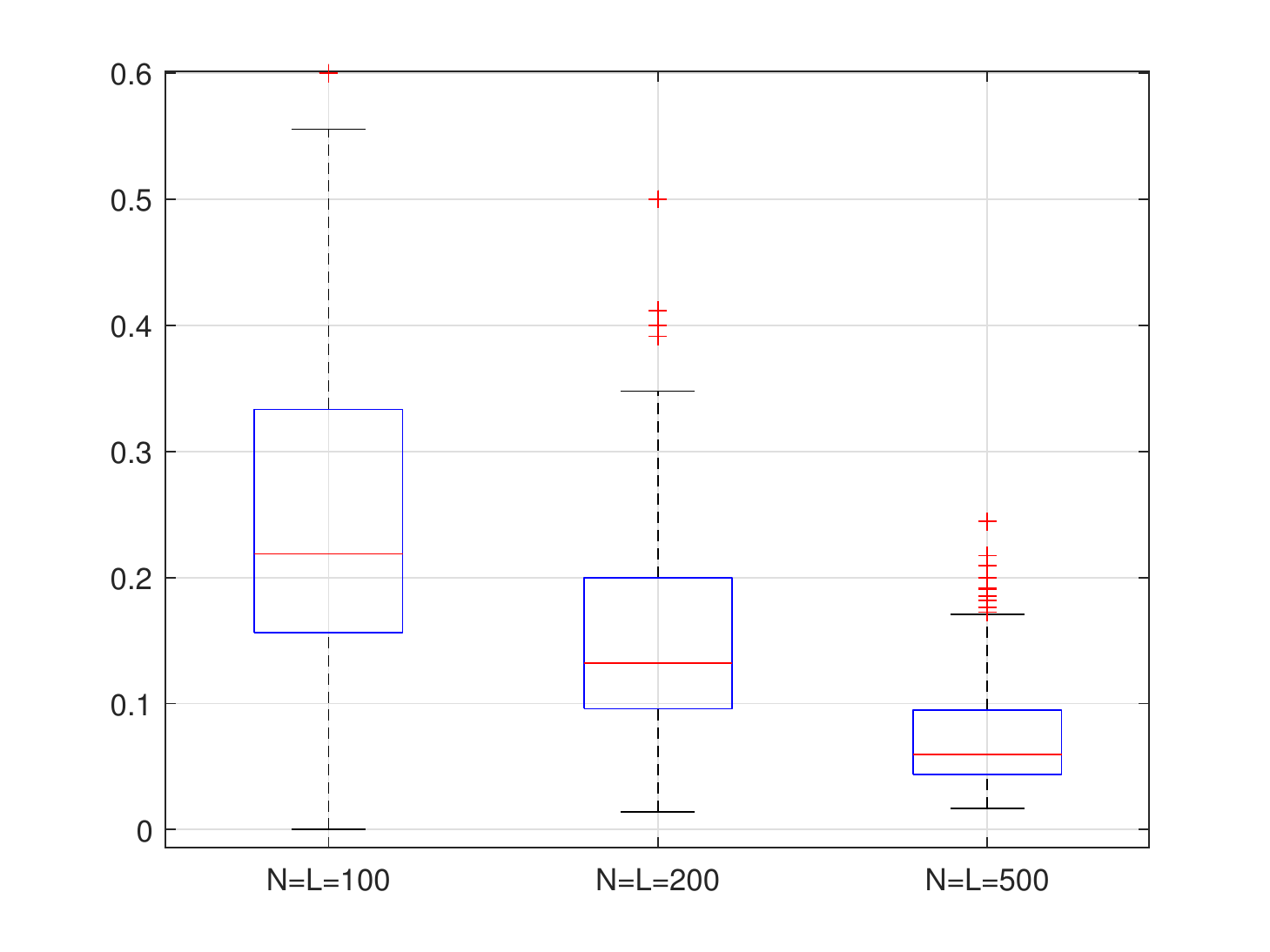}
		\caption{Relative estimation errors of the bandwidth $W$ versus the signal length and the number of cross sections $N=L$ while $\SNR=15$ dB is fixed.}
		\label{fig:W_est_N=M}
	\end{centering}
\end{figure}

We want to comment that preliminary results for the $\arccos$ estimator of one center frequency $\theta$ could be obtained from \eqref{est_f0} but they will not be discussed in depth since much more reliable estimates will be obtained by the subspace method of Sec.~\ref{Subspace}.

\subsection{The Subspace estimator of the band centers}\label{subsec:sim_emp}

Given the estimated rank of the signal covariance matrix in the previous subsection, we proceed to implement the subspace algorithm described at the end of Sec.~\ref{Subspace}.
Again we do a Monte-Carlo simulation of $1000$ trials. The signal length $N=100$ and the $\SNR=15$ dB are fixed, and we change the number of snapshots $L$. The data has already been generated in estimating $W$, and we only need to use   the estimated covariance matrix.


The relative estimation errors of the center frequency is defined as $\|\hat{\theta}-\theta\|/\|\theta\|$, and their values in Monte-Carlo simulations are plotted in Fig.~\ref{fig:sim_MC_two_hid_emp}. It appears that apart from the outliers (the red crosses), the performance of the algorithm is quite good as the cumulative relative error is lower than $2\%$, even in the case of few snapshots ($L=50$). The simulation result also seems to indicate that the algorithm works very well when the covariance estimate is sufficiently accurate. 

\begin{figure}[!h]
	\begin{centering}
		\includegraphics[width=0.5\textwidth]{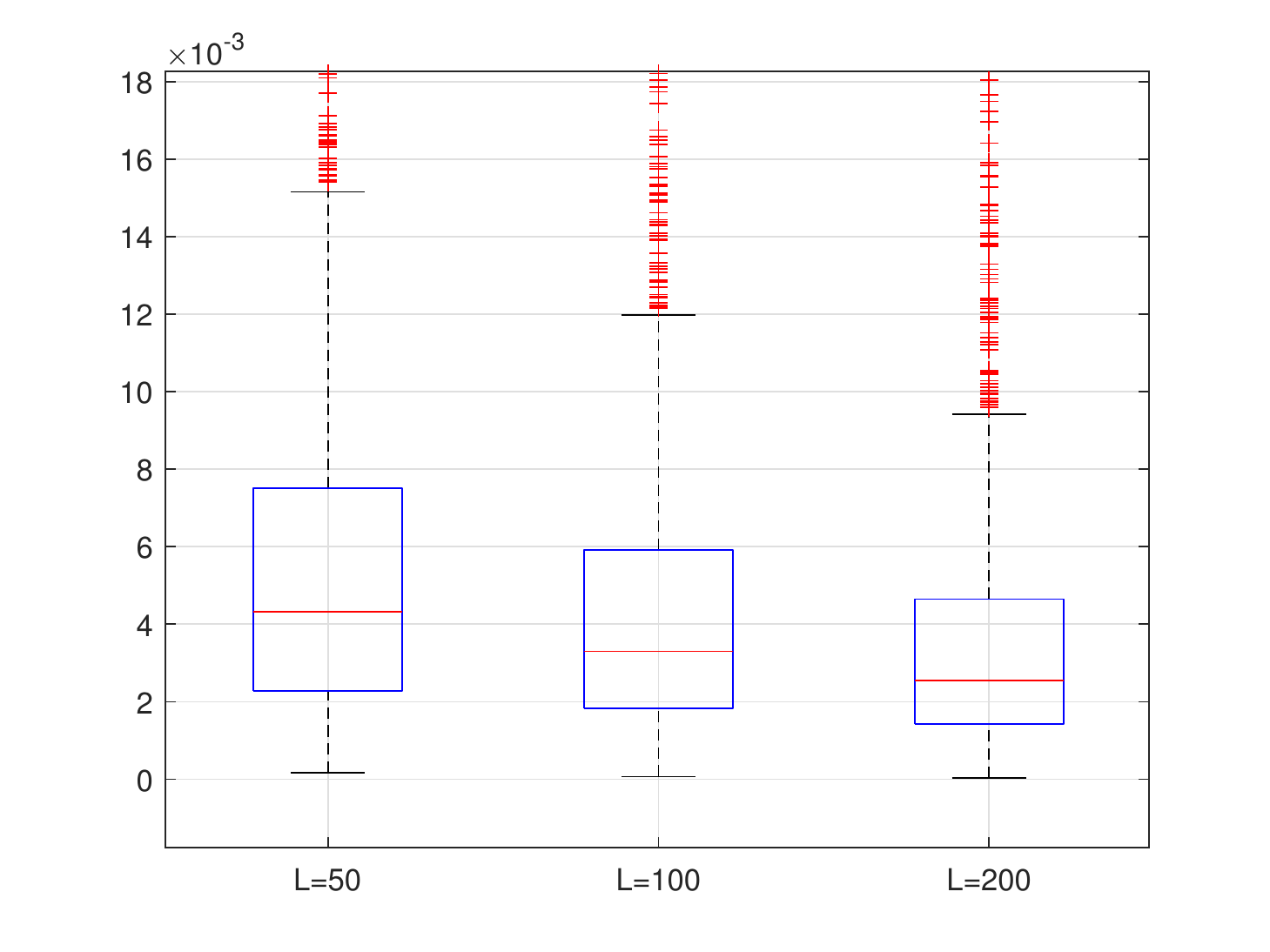}
		\caption{Relative estimation errors of two hidden frequencies $(\theta_1,\theta_2)$ using the Subspace method in Sec.~\ref{Subspace} versus the number $L$ of cross sections with $N=100$ and $\SNR=15$ dB.}
		\label{fig:sim_MC_two_hid_emp}
	\end{centering}
\end{figure}

Fig.~\ref{fig:disc_spec_two_hid} shows the discrete spectrum of the output process \eqref{outpeq} in one simulation trial in the case of $L=100$.
The horizontal axis is scaled to represent the frequency in Hz. In this particular trial, the true hyperparameters are $[\theta_1,\theta_2,W]=2\pi\times[0.1499,0.2524,0.0155]$, and the estimated band centers are $\hat{\theta}=2\pi\times[0.1503,0.2532]$. The theoretical (asymptotic) rank of the signal covariance matrix is $\frac{2\nu W}{\pi}N\approx12$, while the ratio scheme \eqref{rank_substitute} produces a rank estimate equal to $20$. One can see that the Dirac deltas indeed cluster around the true center frequencies inside the supporting interval. 


\begin{figure}[!h]
	\begin{centering}
		\includegraphics[width=0.5\textwidth]{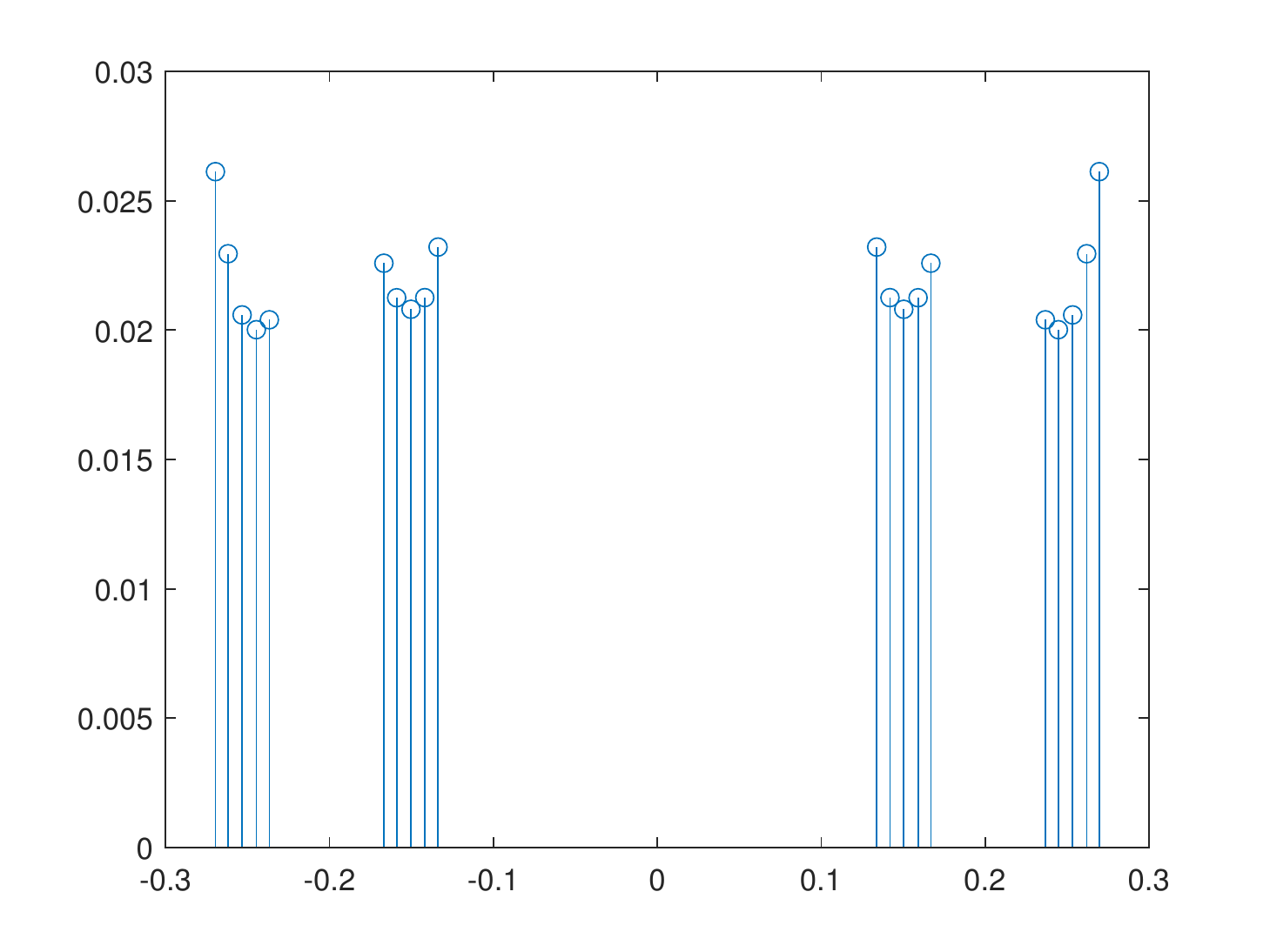}
		\caption{Discrete spectrum estimate with two hidden frequencies. The true hyperparameters are $[\theta_1,\theta_2,W]=2\pi\times[0.1499,0.2524,0.0155]$ and the estimated band centers are $\hat{\theta}=2\pi\times[0.1503,0.2532]$.}
		\label{fig:disc_spec_two_hid}
	\end{centering}
\end{figure}

\subsection{The Bayesian MAP estimator of the frequencies}\label{subsec:sim_Bayesian}

Given the center frequencies and the bandwidth estimated from the Subspace procedure, we can now compute the Empirical Bayes MAP estimator using the algorithm described in Sec.~\ref{Bayes}. The data are the same as those used for covariance and hyperparameter estimation. Once again, we fix the signal length $N=100$ and $\SNR=15$ dB, and do Monte-Carlo simulations of $1000$ trials as the number $L$ of cross sections changes. The relative errors of $\hat{\omega}^{\MAP}$ with respect to the true center frequencies are shown in Fig.~\ref{fig:sim_MC_two_hid_Bayesian}. It appears that the MAP estimate of the frequencies is close to the true band centers with a cumulative relative error below $6\%$. Moreover, the estimation accuracy improves as more snapshots of data are available. It is noticed that the cumulative error size is larger than that of the empirical Subspace  method (Fig.~\ref{fig:sim_MC_two_hid_emp}) probably due to the linearization scheme in solving the original nonlinear least squares problem subject to interval constraints.

\begin{figure}[!h]
	\begin{centering}
		\includegraphics[width=0.5\textwidth]{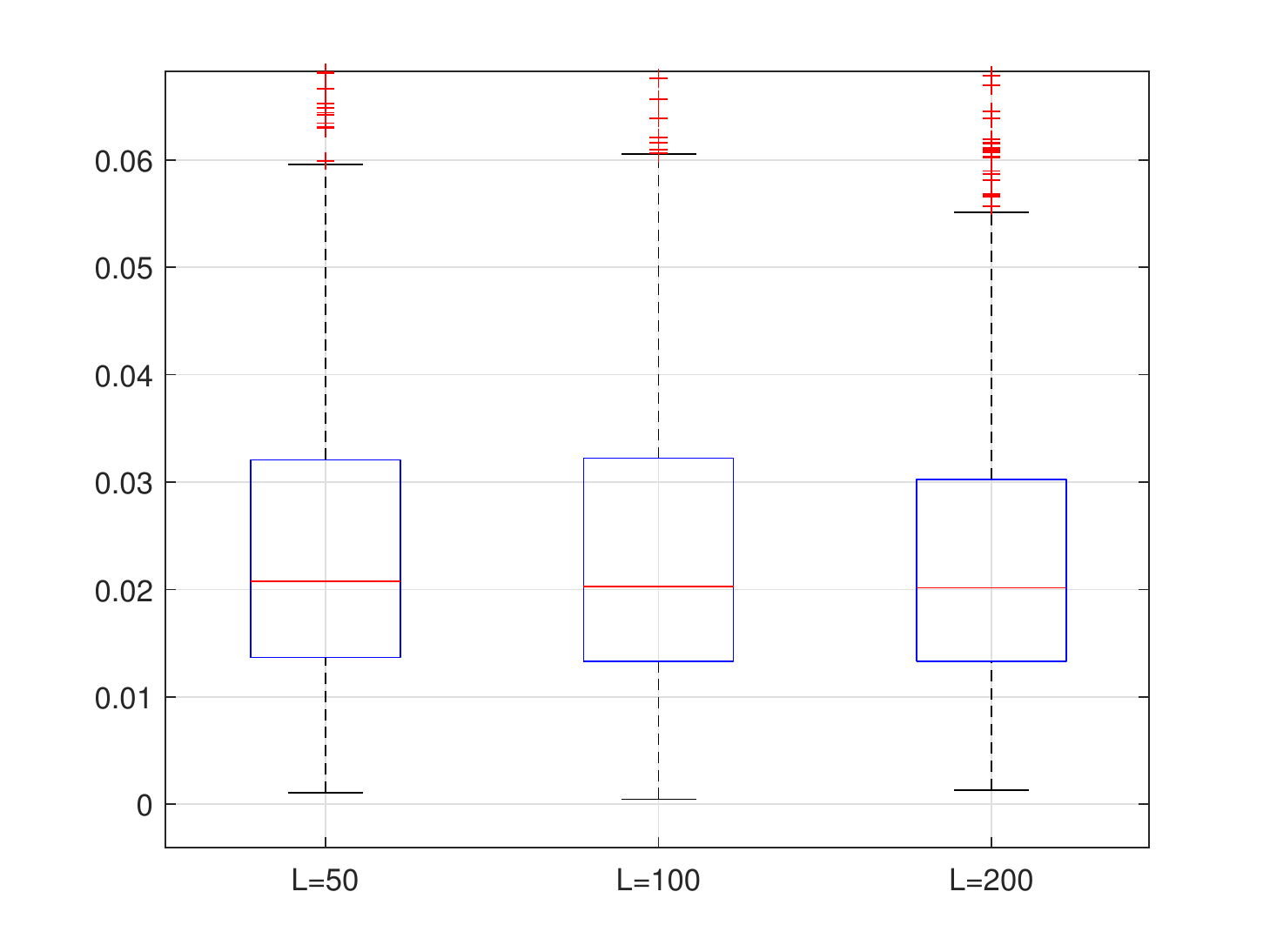}
		\caption{Relative estimation errors of two hidden frequencies $(\theta_1,\theta_2)$ using the Bayesian MAP method in Sec.~\ref{Bayes} versus the number $L$ of cross sections while $N=100$ and $\SNR=15$ dB.}
		\label{fig:sim_MC_two_hid_Bayesian}
	\end{centering}
\end{figure}

\begin{remark} 
A  quite reasonable conjecture, which unfortunately so far we have not   been able to prove rigorously,  is  that  for $N\to \infty$ and small enough $W$, the Bayesian estimate  $\hat \omega^{\MAP}$ should  converges a.s. to the true center frequency $\theta$.

The conjecture is based on the observation that both $\hat \omega^{\MAP}$ and $\hat \theta$ are asymptotic maximizers  of    the likelihood function based on  the same data. In fact, $\hat \theta$ asymptotically solves \eqref{MomentEq} which is the maximizing equation of the   marginal likelihood function,  marginalized by integrating with respect to the a priori  distribution of $\omega$ and hence   parametrized only in terms of the  hyperparameters $(\theta,W)$. We will leave a  detailed discussion of this point  to a future publication.\hfill $\Box$
\end{remark}

\begin{remark}
	At the end of this section want to comment on the difference between  our method and  classical subspace methods for frequency estimation such as MUSIC, ESPRIT, etc. All  classical methods are designed for oscillatory signals with deterministic frequencies and  perform the eigen-truncation of the estimated covariance matrix at an index equal to $2\nu$ where $\nu$ is the number of unknown frequencies (the factor $2$ is due to complexification of the real signal). In contrast, we show that in the case of uniform random frequencies, we have a stochastic multiband signal, and the eigen-truncation should be done at the approximate index $4WN\times\nu$ corresponding to the asymptotic rank of the signal covariance matrix. Based on this observation, it is not surprising that classical subspace methods do not apply to the current problem setup.
	Moreover, our Empirical Bayes procedure provides both the band centers and the bandwidth for the random frequencies, which can be interpreted as confidence intervals for the frequency estimation. \hfill$\Box$
\end{remark}

\begin{remark}
Concerning the Atomic Norm approach, we may just say that it views the observed sinusoidal signal as a deterministic linear combination of elementary exponential components with deterministic frequencies. For this reason (similarly to  the previous remark), it does not seem possible to compare to our random-frequency signal model, although it has been extended to deal with deterministic multiband signals \cite{Zhuetal-17}.\hfill $\Box$
\end{remark}

\section{Conclusions}\label{Conclusions}
We have formulated   the problem of frequency estimation in an Empirical   Bayesian framework by first imposing a natural uniform prior probability density on the unknown frequency. In this way the estimation of the hyperparameters of the a priori distribution can be accomplished by exploiting the special structure of the covariance matrix of the posterior process which has been long studied  in the framework of energy concentration problems by the signal processing community.  In this setting  the solution can be  based on essentially linear techniques  of  subspace identification. Using the estimated prior parameters one can adapt the prior to the data and this  leads to Bayesian estimates which are asymptotically maximum likelihood and therefore the best  possible in a variety  of metrics.
The simulation results using this    Empirical Bayesian philosophy     are very encouraging.

\appendix

In the proof of Theorem \ref{thm_asymp_eigen} we shall need two auxiliary lemmas. The first is just a simple technical fact.

\begin{lemma}\label{lem_lim}
If two sequences of bounded real numbers $\{a_n\},\{b_n\}$ are such that
\begin{equation}\label{lim_diff_zero}
\lim_{n\to\infty} (a_n-b_n) = 0,
\end{equation}
then
\begin{equation*}
\limsup_{n\to\infty} a_n = \limsup_{n\to\infty} b_n ,
\end{equation*}
and
\begin{equation*}
\liminf_{n\to\infty} a_n = \liminf_{n\to\infty} b_n .
\end{equation*}
\end{lemma}

\begin{proof}
The argument is quite standard. Let $\bar{a}:=\limsup_{n\to\infty} a_n$. Then there exits a subsequence $\{a_{n_k}\}$ converging to $\bar{a}$. Define $\hat{b}:=\limsup_{k\to\infty} b_{n_k}$. Then there exists a sub-subsequence $\{b_{n_{k_j}}\}$ converging to $\hat{b}$. The condition \eqref{lim_diff_zero} holds for the subsequence indexed by $n_{k_j}$, which implies that $\bar{a}=\hat{b}$. It then follows that $\bar{b}:=\limsup_{n\to\infty} b_n \geq \hat{b}=\bar{a}$. A symmetric argument leads to $\bar{a}\geq\bar{b}$, and therefore $\bar{a}=\bar{b}$. The proof for the limit inferior is similar and hence omitted.
\end{proof}

The next lemma concerns the sum of squared eigenvalues of $\Rb$.

\begin{lemma}\label{lem_sum_sqr_eigen}
\begin{equation}
\lim_{N\to\infty} \frac{1}{N} \sum_{j=1}^{N} \lambda_j^2(N) = \frac{\m(J)}{2\pi} \,.
\end{equation}
\end{lemma}
\begin{proof}
Since $\Rb$ is Hermitian, we have
\begin{equation}
\begin{split}
\sum_{j=1}^{N} \lambda_j^2(N) & = \trace \Rb^2 =\trace (\Rb \Rb^*) \\
 & = \sum_{j=-N+1}^{N-1} |\rho(j)|^2 (N-|j|).
\end{split}
\end{equation}
It follows that
\begin{equation}
\frac{1}{N} \sum_{j=1}^{N} \lambda_j^2(N) = \sum_{j=-N+1}^{N-1} |\rho(j)|^2 \left( 1-\frac{|j|}{N} \right).
\end{equation}
We can view the latter summation over $\Zbb$ by adding zeros. Apparently, each term in the infinite sum is dominated by $|\rho(j)|^2$. Moreover, for each fixed $j$ the term-wise limit as $N\to\infty$ is also $|\rho(j)|^2$. Applying Lebesgue's dominated convergence theorem for the counting measure on $\Zbb$, we can conclude that
\begin{equation}
\begin{split}
\lim_{N\to\infty} \frac{1}{N} \sum_{j=1}^{N} \lambda_j^2(N) & = \sum_{j\in\Zbb} |\rho(j)|^2 \\
 & = \frac{1}{2\pi} \int_{-\pi}^{\pi} |\chi_J(\omega)|^2 \d\omega = \frac{\m(J)}{2\pi},
\end{split}
\end{equation}
where the second equality is the Parseval identity.
\end{proof}

\noindent{\bf Proof of Theorem \ref{thm_asymp_eigen}}
\begin{proof} 
We first show that the number of eigenvalues not close to $0$ or $1$ is $o(N)$.
To this end, define the function
\begin{subequations}
\begin{align}
\Jbb(N) & := \sum_{j=1}^{N} \lambda_j(N) \left( 1-\lambda_j(N) \right) \label{J_n_func}\\
 & = \sum_{j=1}^{N} \lambda_j(N) - \sum_{j=1}^{N} \lambda_j^2(N) ,
\end{align}
\end{subequations}
where each summand in \eqref{J_n_func} is nonnegative.
Then according to \eqref{sum_eigen} and Lemma \ref{lem_sum_sqr_eigen}, we have
\begin{equation}
\lim_{N\to\infty} \frac{\Jbb(N)}{N} = 0 .
\end{equation}
In other words, the function $\Jbb(N)$ is $o(N)$.
Fix $0<\delta<\gamma<1$, and the number of eigenvalues $\delta\leq \lambda_j(N) <\gamma$ is $M(\delta,n)-M(\gamma,N)$. Clearly, for these eigenvalues we have
\begin{equation}
\lambda_j(N) \left( 1-\lambda_j(N) \right) > \delta (1-\gamma) :=\nu>0,
\end{equation}
which implies that
\begin{equation}
\begin{split}
\Jbb(N) & \geq \sum_{\delta\leq \lambda_j(N) <\gamma} \lambda_j(N) \left( 1-\lambda_j(N) \right) \\
 & \geq \nu \left[ M(\delta,N)-M(\gamma,N) \right] \geq 0.
\end{split}
\end{equation}
It follows 
that
\begin{equation}\label{lim_zero}
\lim_{N\to\infty} \frac{M(\delta,N)-M(\gamma,N)}{N} = 0,
\end{equation}
which means that the quantity $M(\delta,N)-M(\gamma,N)$ is also $o(N)$.

Next, define the quantities
\begin{equation}
\begin{split}
M_+ & := \limsup_{N\to\infty} \frac{M(\gamma,N)}{N} , \\
M_- & := \liminf_{N\to\infty} \frac{M(\gamma,N)}{N} . \\
\end{split}
\end{equation}
Applying Lemma~\ref{lem_lim} in this appendix to the relation \eqref{lim_zero}, we know that both $M_+$ and $M_-$ do not depend on $0<\gamma<1$. We want to establish that the two quantities coincide so that the ordinary limit in \eqref{lim_asymp_eigen} exits and is equal to the common value. Observe that
\begin{equation}
\trace \Rb = \sum_{j=1}^{M(\gamma,N)} \lambda_j(N) + \underbrace{\sum_{M(\gamma,N)+1}^{N} \lambda_j(N)}_{> 0} \geq \gamma M(\gamma,N),
\end{equation}
and similarly
\begin{equation}
\begin{split}
\sum_{j=1}^{N} \lambda_j^2(N) & = \sum_{j=1}^{M(\gamma,N)} \lambda_j^2(N) + \underbrace{\sum_{M(\gamma,N)+1}^{N} \lambda_j^2(N)}_{\text{here each } \lambda_j(N)<\gamma} \\
 & < \sum_{j=1}^{M(\gamma,N)} 1 + \sum_{M(\gamma,N)+1}^{N} \gamma \lambda_j(N) \\
 & < M(\gamma,N) + \gamma \trace \Rb .
\end{split}
\end{equation}
It follows that
\begin{equation}
\sum_{j=1}^{N} \lambda_j^2(N) - \gamma \trace \Rb \leq M(\gamma,N) \leq \frac{\trace \Rb}{\gamma},
\end{equation}
and furthermore, we have
\begin{subequations}
\begin{align}
M_+ & \leq \limsup_{N\to\infty} \frac{\trace \Rb}{\gamma N} = \frac{\m(J)}{2\pi\gamma} , \label{M+_inequal} \\
M_- & \geq \liminf_{N\to\infty} \frac{1}{N} \left( \sum_{j=1}^{N} \lambda_j^2(N) - \gamma \trace \Rb \right) \notag \\
 & = (1-\gamma) \frac{\m(J)}{2\pi} \label{M-_inequal} ,
\end{align}
\end{subequations}
where we have used Lemma \ref{lem_sum_sqr_eigen} again in \eqref{M-_inequal}.
Letting $\gamma\to1$ in \eqref{M+_inequal} and $\gamma\to0$ in \eqref{M-_inequal}, we obtain
\begin{equation}
\frac{\m(J)}{2\pi} \leq M_- \le M_+ \le \frac{\m(J)}{2\pi},
\end{equation}
and the claim of the theorem follows.
\end{proof}

The next proposition concerns the time average of one sample path of the noisy sinusoidal signal.

\begin{proposition}\label{prop_time_average}
	Let
	\begin{equation}
	\begin{split}
	y(t) & = x(t) + w(t) \\
	 & = a \cos (\omega t) + b \sin (\omega t) + w(t)
	\end{split}
	\end{equation}
	be a sample path of the process \eqref{y_measurement}, where $t=1,2,\dots$. Then for each fixed $\omega$ with $|\omega|<\pi$,
	\begin{equation*}
	\begin{split}
	\frac{1}{N} \sum_{t=1}^{N} y(t+\tau) y(t) \to \frac{a^2+b^2}{2} \cos \omega\tau + \sigma_{\wb}^2 \delta(\tau,0)
\end{split}
\end{equation*}
	as $N\to\infty$ with probability one.
\end{proposition}

\begin{proof}
We have
\begin{align*}
	  \frac{1}{N} \sum_{t=1}^{N} y(t+\tau) y(t) &= \frac{1}{N} \sum_{t=1}^{N} \left[ x(t+\tau) x(t) + x(t+\tau) w(t) \right. \\ 
	&  \left. + w(t+\tau) x(t) + w(t+\tau) w(t) \right]
\end{align*}
and that the first time average converges to $\frac{a^2+b^2}{2} \cos \omega\tau$ is shown in \cite[pp. 105-109]{Soderstrom-S-89} or \cite[pp. 171-172]{Stoica-M-05}. That the average of each cross term in the middle tends to $0$, follows since the process $\tilde \wb(t):= e^{i\omega t} \wb(t)$ is (complex) zero-mean i.i.d. and by the  assumed uncorrelation so  is also  $\ab e^{i\omega \tau}\tilde \wb(t)$ and hence so is  its real part, so that the law of large numbers holds for each cross term. The time average of the last term tends to $\sigma_{\wb}^2 \delta(\tau,0)$ again by the law of large numbers. \end{proof}


\bibliographystyle{IEEEtran}
\bibliography{biblio-FreqEst}




\end{document}